\documentclass[12pt, draftclsnofoot, onecolumn]{IEEEtran}
\usepackage{graphicx,indentfirst}
\usepackage{booktabs}
\usepackage{indentfirst}
\usepackage{amsmath}
\usepackage{cite}
\usepackage{multirow}
\usepackage{subcaption}
\usepackage{makecell}
 \usepackage{mathrsfs}
\usepackage{amsmath}
\usepackage[english]{babel}
\usepackage{amsthm}
\usepackage{multirow}
\usepackage{algorithm, algorithmic}
\usepackage{paralist}
\usepackage{amsthm}
\usepackage{amssymb}
\usepackage{amsfonts}
\usepackage{color}
\newtheorem{proposition}{Proposition}
\newtheorem{theorem}{Theorem}
\newtheorem{lemma}{Lemma}
\newtheorem{assumption}{Assumption}

\newtheorem{definition}{Definition}

\IEEEoverridecommandlockouts

\begin{document}
\bibliographystyle{IEEE2}
\title{A Stackelberg Game Approach Towards Socially-Aware Incentive Mechanisms for Mobile Crowdsensing (Online report)}
\author{Jiangtian~Nie, Jun~Luo,~\IEEEmembership{Senior Member,~IEEE,} Zehui~Xiong, Dusit~Niyato,~\IEEEmembership{Fellow,~IEEE,} Ping~Wang,~\IEEEmembership{Senior Member,~IEEE}
\thanks{Jiangtian Nie is with ERI@N, Interdisciplinary Graduate School as well as School of Computer Science and Engineering, Nanyang Technological University, Singapore. Jun Luo, Zehui Xiong, Dusit Niyato and Ping Wang are with School of Computer Science and Engineering, Nanyang Technological University, Singapore.\vspace*{-4mm}}}
\maketitle
\vspace*{-12mm}
\begin{abstract}
Mobile crowdsensing has shown a great potential to address large-scale data sensing problems by allocating sensing tasks to pervasive mobile users. The mobile users will participate in a crowdsensing platform if they can receive satisfactory reward. In this paper, to effectively and efficiently recruit sufficient number of mobile users, i.e., participants, we investigate an optimal incentive mechanism of a crowdsensing service provider. We apply a two-stage Stackelberg game to analyze the participation level of the mobile users and the optimal incentive mechanism of the crowdsensing service provider using backward induction. In order to motivate the participants, the incentive is designed by taking into account the social network effects from the underlying mobile social domain. For example, in a crowdsensing-based road traffic information sharing application, a user can get a better and accurate traffic report if more users join and share their road information. We derive the analytical expressions for the discriminatory incentive as well as the uniform incentive mechanisms. To fit into practical scenarios, we further formulate a Bayesian Stackelberg game with incomplete information to analyze the interaction between the crowdsensing service provider and mobile users, where the social structure information (the social network effects) is uncertain. The existence and uniqueness of the Bayesian Stackelberg equilibrium is validated by identifying the best response strategies of the mobile users. Numerical results corroborate the fact that the network effects tremendously stimulate higher mobile participation level and greater revenue of the crowdsensing service provider. In addition, the social structure information helps the crowdsensing service provider to achieve greater revenue gain.
\end{abstract}
\vspace*{-3mm}

\section{Introduction}\label{Sec:Introduction}
In the past decade, we have been witnessing a fast proliferation of mobile users and devices in daily life. The ubiquitous mobile devices with various embedded functional sensors have remarkably promoted the information generation process. These advances stimulate the rapid development of mobile sensing technologies, and mobile crowdsensing becomes one of the most attractive and popular paradigms. Mobile crowdsensing leverage the sensing capacity of worldwide available smart phones, e.g., GPS, camera and digital compass, to collect distributed sensory data.

The basic crowdsensing platform typically includes a cloud-based system and a collection of smart phones or mobile users. The platform can post a set of sensing tasks with different purposes, and mobile users are actively involved to perform the corresponding tasks. Realizing the great business potential, lots of crowdsensing-based applications have been designed and introduced in a number of areas. Sensorly~\cite{Sensorly} is dedicated for WiFi coverage information, Waze~\cite{Waze} and GreenGPS~\cite{GreenGPS} are to collect road traffic information, DietSensor~\cite{DietSensor} is proposed to share and track the diet nutrition, and Noisetube~\cite{noisetube} is for monitoring the noise pollution.

Nevertheless, voluntary participation in the crowdsensing platform may not be sustainable. This is from the fact that the mobile users need to spend their own resources, e.g., smart phone battery, CPU computing power, storage memory, to accomplish the sensing tasks. Another major concern that discourages the mobile users from participation comes from the potential privacy issues. Therefore, individuals are reluctant to participate and share their collected information due to the lack of sufficient motivation and incentive. Nevertheless, the crowdsensing systems heavily rely on total user participation level and the individual contribution from each user. To stimulate and recruit users with mobile devices to participate in crowdsensing, the crowdsensing platform administrator, i.e., the Crowdsensing Service Provider (CSP), usually provides a reward for mobile users as a monetary incentive to compensate their cost or risk.

It is challenging to design the incentive mechanism that achieves a sustainable and profitable market for the CSP. When the reward is small, the collected sensing information from mobile participants is insufficient. Conversely, when the reward is large, the CSP may incur excessive operation cost. Accordingly, an efficient incentive mechanism has become an emerging topic of interest for a large number of researchers. However, most of the existing works have addressed the incentive mechanism for mobile crowdsensing without considering the interdependent behaviors of mobile users from social domain. This interdependency originates from the network effects. Traditionally, \textit{network effects} refer to the phenomenon that public goods or service is more valuable if it is adopted by more users. In crowdsensing service, the participation behavior of mobile users can be deemed as buying ``public goods'', which means that the mobile users are more willing to participate if the number of other users is greater. For example, in a crowdsensing-based road traffic information sharing application, a user can receive a better and accurate traffic report if more users join and share their road information. Consequently, the complex and interdependent user behaviors post a remarkable challenge to the operation of the crowdsensing platform. More importantly, the network effects frequently exist in densely connected social relationships, which is one of the key criteria to promote the wisdom of crowds~\cite{surowiecki2007wisdom, chakeri2017incentive}.

Nevertheless, only a few works~\cite{chen2016incentivizing, zhang2016privacy} have studied the incentive mechanism for crowdsensing and exploited the network effects at the same time. The authors in~\cite{chen2016incentivizing, zhang2016privacy} investigated the behaviors of mobile users under the global network effects\footnote{Global network effects refer to as the phenomenon that a user will obtain higher value when its behavior aligns with any other users~\cite{easley2010networks}.}, which is not appropriate for the structure of an underlying social domain. By contrast, social (local) network effects refer to the case where each user is only influenced directly by the decisions of other densely socially-connected users~\cite{candogan2012optimal, xiong2017sponsor}. For example, in a mobile crowdsensing platform for sharing road traffic information, a user (driver) can get a better route if more neighbourhood users (users in the same or nearby road) of this user join and contribute their traffic data~\cite{Waze, GreenGPS}. On the contrary, this user cannot obtain any benefits if the users in other distant roads join and share their traffic data. This fact motivates us to explore the role of social (local) network effects in designing the incentive mechanism of crowdsensing service.

In this paper, we propose novel incentive mechanisms by leveraging the underlying social network effects to attract participants to crowdsensing platform. First, the crowdsensing platform administrator, i.e., the CSP, determines the incentive, i.e., the offered reward, to maximize its revenue. Then, based on the given reward, the mobile users decide on their participation level individually by taking the social network effects into account. The above rewarding and participating decision marking process can be inherently modeled as a hierarchical Stackelberg game. Moreover, we consider the uncertainty of social network effects, which is commonly applicable to some of the real-world crowdsensing applications. As such, we also formulate the Bayesian Stackelberg game with incomplete information to analyze and evaluate the impacts of uncertainty of social network effects. The major contributions of this paper are summarized as follows:
\begin{itemize}
 \item To our best knowledge, this is the first work on designing incentive mechanisms for mobile crowdsensing with the consideration of complete and incomplete information on social network effects. In particular, we exploit the social network effects in the game model, which utilizes the structural properties from the underlying social domain, and fully characterizes the social relations among the mobile users.
 \item We model the interaction between the CSP and mobile users as a two-stage Stackelberg game and analyze each stage systematically through backward induction. We investigate two types of incentive mechanism for the crowdsensing platform with complete and incomplete information on social network effects, i.e., the Stackelberg game based incentive mechanism and the Bayesian Stackelberg game based incentive mechanism, respectively.
 \item In the Stackelberg game based incentive mechanism, the CSP and mobile users acquire the exact information on underlying social network effects. We propose the optimal incentive mechanism in terms of discriminatory incentive and uniform incentive, in which the CSP offers the different or the same reward to all the mobile users. For both, we are able to obtain the analytical expression for optimal reward. 
 \item In the Bayesian Stackelberg game based incentive mechanism, the specific information on social network effects is under uncertainty. We obtain a unique Bayesian Nash equilibrium adopted by the mobile users in closed-form. Thereafter, the existence and uniqueness of the Bayesian Stackelberg equilibrium is proved by identifying the best response strategies of the mobile users.
 \item Performance evaluation is provided to demonstrate the effectiveness of the proposed game theory based socially-aware incentive mechanisms. Numerical results show that the network effects play an important role to promote higher participation level and thus greatly improve the revenue of the CSP. Moreover, the information about social relationship, i.e., social structure, helps the CSP to achieve greater revenue gain.
\end{itemize}

The rest of this paper is organized as follows. Section~\ref{Sec:Related} provides the literature review. Section~\ref{Sec:Model} describes the system model and the game formulation. In Section~\ref{Sec:Solution}, we analyze the mobile user participation level and optimal reward using backward induction. In Section~\ref{Sec:Bayesian}, we formulate a Bayesian Stackelberg game where the social structure information is uncertain, and study the Bayesian game equilibrium. Section~\ref{Sec:Simulation} presents the performance evaluation and Section~\ref{Sec:Conclusion} concludes the paper.

\section{Related works}\label{Sec:Related}

Recently, a large number of prior works have been dedicated to designing incentive mechanisms~\cite{zhang2016incentives}. Auction is a widely-adopted method to design the incentive mechanisms. In~\cite{koutsopoulos2013optimal}, the authors presented a mechanism for participation level determination and reward allocation using optimal reverse auction, in which the CSP receives service queries and initiates an auction for user participation. The authors in~\cite{xu2015incentive} explored the truthful mechanism with strong requirements of data integrity where the tasks are time window dependent. A reverse auction framework is adopted to derive the optimal incentive, which is computationally efficient and individually rational. In~\cite{han2016truthful}, the authors investigated scheduling problem, where the CSP announces a set of tasks and then mobile users compete for the tasks based on the sensing costs and available time periods. The approximation mechanisms for the CSP to schedule and reward the users under certain budget is provided. The authors in~\cite{zhou2017truthful} studied incentivizing user participation and assigning location dependent tasks with capacity budget. A truthful one-round auction with approximation algorithm is proposed to obtain the optimal reward offered to the participants. In~\cite{yang2016incentive}, the authors considered the user-centric model where each user can ask for reserve price, and designed the truthful and scalable auction mechanism for the CSP to achieve revenue maximization. The authors in~\cite{zheng2017budget} addressed how to maximize the valuation of the covered interested regions under limited budget for strategy-proof mobile crowdsensing. In~\cite{wang2018melody}, the authors proposed a long-term dynamic incentive mechanism to capture the dynamic nature of long-term data quality of participants, where a truthful, quality-aware and budget feasible algorithm is designed for task allocation with polynomial-time computational complexity. The authors in~\cite{lin2017frameworks} investigated the auction based incentive mechanism considering social cost minimization and privacy preservation. The participants are selected based on predefined score functions by the CSP, and the computational efficiency, individual rationality, truthfulness and differential privacy are guaranteed. To prevent the Sybil attack where a user illicitly disguises other identities to obtain benefits, the authors in~\cite{lin2017sybil} designed Sybil-proof auction-based incentive mechanisms.

In addition to auction mechanism design, the incentive mechanisms are examined with different objectives. For example, the authors in~\cite{zhan2017incentive} considered that the sensing information has an attached time-sensitive value that decreases over time and focused on the incentive design for cooperative data collection of participants. In~\cite{chakeri2017incentive}, the authors explored the incentive mechanism with multiple CSPs, where the incentive mechanism is modeled as a noncooperative game. The discrete time dynamic inspired by the best response dynamics is proposed to achieve the Nash equilibrium of the modeled game. The authors in~\cite{gan2017social} presented a novel Vickrey-Clarke-Groves game based incentive mechanism for sensing resource sharing by the encouraged participants. The task allocation and resource sharing algorithm is developed to achieve the social fairness and efficiency tradeoff. The authors developed a new framework called Steered Crowdsensing in~\cite{kawajiri2014steered}, which controls incentives by introducing gamifications with monetary reward to location-based services. In~\cite{peng2015pay}, the authors incorporated the consideration of data quality into the mechanism, and rewarded the participant depending on the quality of its collected data. The authors in~\cite{luo2015crowdsourcing} applied Tullock contests to design incentive mechanisms, where the reward includes a fixed contest prize, and Tullock prize function depending on the winner's contribution. In~\cite{duan2012incentive}, the authors proposed a reward-based collaboration mechanism, where the CSP announces a total reward to be shared among collaborators, and the task and reward are allocated if sufficient number of participants are willing to collaborate. In~\cite{han2016posted}, the authors studied a quality-aware Bayesian incentive problem for robust crowdsensing, where the data quality and sensing cost of users are drawn from known distribution.

In~\cite{chakeri2017iterative}, the authors considered a sealed market for the CSP, where the participants have imperfect information on other participants behavior. The iterative game framework is introduced and the incentive mechanism is obtained by best response dynamics with several iterations. The authors in~\cite{zhanyu2017incentive} formulated the one-to-many Nash bargaining game to model the interaction between the CSP and participants. The distributed algorithm that ensures the participators' privacy and reduces the computation load of the CSP is provided. The authors in~\cite{wang2018blockchain} proposed blockchain based distributed incentive mechanism which can remove the security threats caused by a ``trustful'' crowdsensing center. The participants with sensing information contribution obtain the reward that is recorded in transaction blocks. In our previous work~\cite{nie2017socially}, we considered the social network effects that promote the participation level while designing the incentive mechanism. In~\cite{li2017dynamic}, the authors also pointed out the importance of ``network effects'' on social information sharing with the problem of dynamic routing. For example, a user traveling on one route benefits from the content collected by users traveling on another route.
However, the scenario where the network effects is certain has its limitation which may not be applicable to some of the real-world applications such as crowdsensing. In this paper, we consider the uncertain scenario where the social structure information is not known exactly by the CSP and participants. 
\begin{figure}[t]
\centering
\includegraphics[width=.75\textwidth]{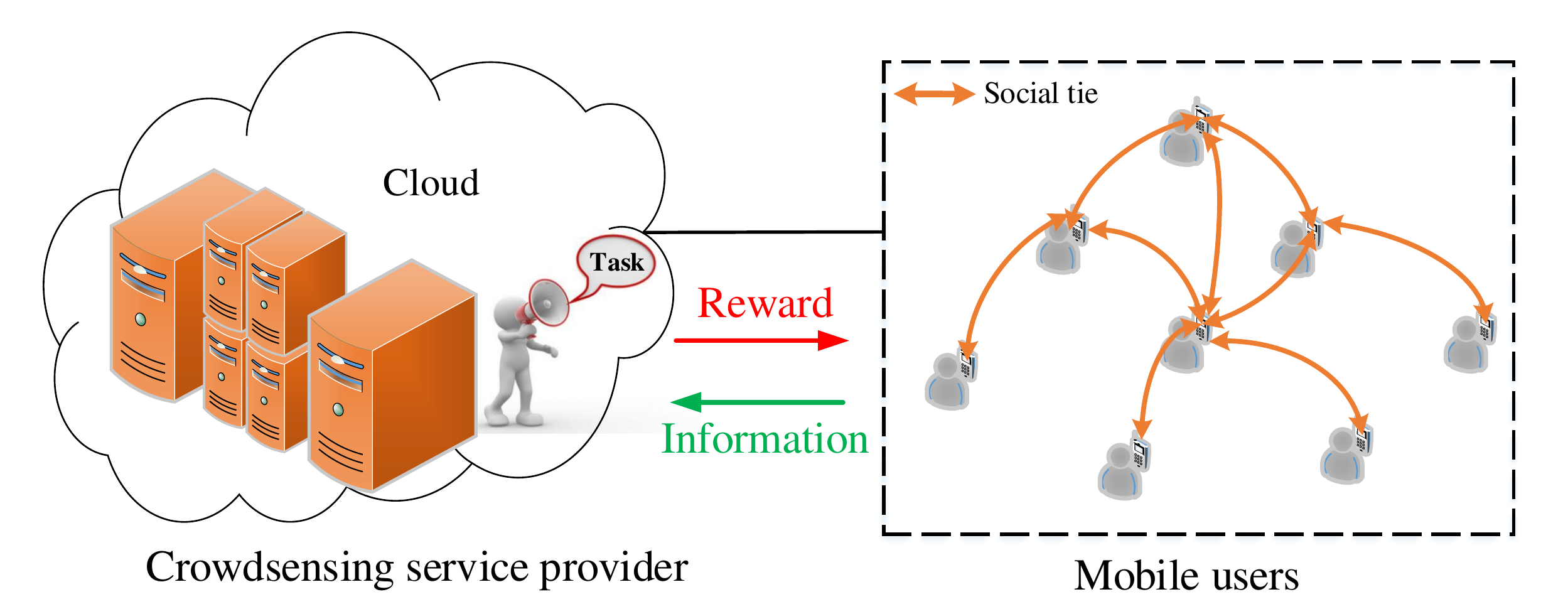}
\caption{Basic system model of mobile crowdsensing platform with social network effects.}\label{Fig:Model}
\end{figure}

\section{System description and game formulation}\label{Sec:Model}
\begin{table}[!]\label{definition}
\centering
\caption{Main Notations}
\begin{tabular}{c|l}
\hline
Symbol & Definition \\
\hline
$\cal N$, $N$	&	Set of MUs, and the total number of MUs, respectively	\\
\hline
${\cal N}_i$	&	Set of social neighbours of MU $i$ \\
\hline
$x_i$ & Participation level of MU $i$, i.e., the effort level in participation\\
\hline
$\mathbf{x}$, $\mathbf{x}_{-i}$ & The participation levels of all the MUs and all other MUs except MU $i$, respectively\\
\hline
$r_i$ & The offered reward to MU $i$ from the CSP\\
\hline
$a_i, b_i$ & The coefficients capturing the intrinsic value of different MUs\\
\hline
$g_{ij}$ & The influence of MU $j$ on MU $i$\\
\hline
$c$ &  The MU's unit cost associated to its participation level\\
\hline
$\mu$ & The parameter representing the equivalent monetary worth of MUs' participation level\\
\hline
$s, t$ & The coefficients capturing the concavity of the profit obtained from the total contribution of all MUs\\
\hline
$u_i$ & The utility of MU $i$ \\
\hline
$\Pi$ & The revenue of the CSP\\
\hline
$\gamma$ & Given social network effects coefficient\\
\hline
$k$ & The out-degree of MU\\
\hline
$l$ & The in-degree of MU\\
\hline
$P(k)$ & The out-degree distribution of MUs\\
\hline
$H(l)$ & The in-degree distribution of MUs\\
\hline
${\rm Avg}({{\bf{x}}_{ - {{i}}}})$ & The average participation level of social neighbours of MU $i$\\
\hline
$\bar k$ & Average level of social network effects\\
\hline
${\sigma ^2}_k, {\sigma ^2}_l$ &  The two variance of out-degree and in-degree distributions \\
\hline
\end{tabular}
\end{table}
We model the interaction among the Crowdsensing Service Provider (CSP) and the socially-aware participants, i.e., Mobile Users (MUs), as a hierarchical Stackelberg game, where the action of each MU is to choose an individual participation level and the action of the CSP is to give the payment as a reward to incentivize the MUs (Fig.~\ref{Fig:Model}). Consider a set of MUs denoted by ${\cal N} \buildrel \Delta \over = \{ 1, \ldots ,N\}$. Each MU $i \in {\cal N}$ determines its participation level or effort level, e.g., sensing data transmission frequency or sensing resolution, denoted by $x_i$ where ${x_i} \in (0,+ \infty)$.

Let $\mathbf{x} \buildrel \Delta \over = ({x_1}, \ldots ,{x_N})$ and $\mathbf{x}_{-i}$ denote the participation levels of all the MUs and all other MUs except MU $i$, respectively. The reward per effort unit provided to the MUs is given as: ${\bf{r}} = {[{{r}_{1}}, \ldots {{r}_{i}}, \ldots ,{{r}_{N}}]^\top}$. Then, the utility of MU $i$ is given by
\begin{equation}\label{Eq:1}
{u_i}({x_i},{{\bf{x}}_{ - {{i}}}}) = {f_i}({x_i}) + {\Phi}({x_i},{{\bf{x}}_{ - {{i}}}}) + {r(x_i)} - {c}({x_i}).
\end{equation}
The first term ${f_i}(x)$ represents the private utility or internal effects that MU $i$ obtains from the participation, which can be formulated as ${f_i}(x_i) = {a_i}{x_i} - {b_i}{x_i}^2$, where $a_i > 0$ and $b_i > 0$ are the coefficients that capture the intrinsic value of the participation to different MUs with heterogeneity~\cite{candogan2012optimal, xiong2017sponsor}. For example, in a crowdsensing-based traffic information sharing application, when a user reports speed and location on a certain road more frequently, i.e., larger $x_i$, the accuracy of the traffic condition on that road is higher~\cite{Waze, GreenGPS}. As in~\cite{candogan2012optimal}, the quadratic form of the internal utility not only allows for tractable analysis, but also serves as a good second-order approximation for a broad class of concave utility functions. Additionally, the linear-quadratic function captures the decreasing marginal returns from participation. In particular, $a_i$ models the maximum internal participation willingness rate, and $b_i$ models such willingness elasticity factor.

The second term, $\Phi({x_i},\mathbf{{x}}_{-i})$ denotes the external benefits gained from the network effects, which is the key component from Eq.~(\ref{Eq:1}). In crowdsensing applications, an MU can enjoy an additional benefit from information contributed or shared by the others~\cite{easley2010networks}. The existing work explored the network effects of global nature, where the additional benefits due to new coming MUs are the same for all the existing MUs~\cite{chen2016incentivizing}. However, due to the structural properties from the underlying social domain, it is more appropriate to consider the network effects locally in crowdsensing service, i.e., the social network effects. Then, we introduce the adjacency matrix ${\mathbf G} = {[{g_{ij}}]_{i, j \in \mathcal{N}}}$. The elements in matrix $g_{ij}$ indicates the influence of MU $j$ on MU $i$, which can be unidirectional or bidirectional. For example, with a larger $g_{ij}$, the participation level of MU $j$ can increase the utility of MU $i$ faster. Motivated by the idea of social reciprocity~\cite{fehr2000fairness, chen2013social}, a user's social behavior to another is likely to imitate and be imitated by the latter's behavior to the former. As a result, two MU social ties to each other tend to be the same. Thus, we consider $g_{ij}= g_{ji}$ in this paper, i.e., the social tie is reciprocal. Nevertheless, the proposed model can be applied to asymmetric social ties straightforwardly (Please refer to Appendix \ref{Sec:discussion on gij}). Specifically, we adopt ${\sum _{j \in \cal N}}{g_{ij}}{x_i}{x_j}$ to represent the additional benefits obtained from the network effects, similar to that in~\cite{candogan2012optimal, gong2015network, xiong2017sponsor}.

The third term, ${r(x_i)}$, is the reward from CSP to the MU $i$, which is equal to ${r_i}{x_i}$ , i.e., the reward is a linear function to the effort or participation level. The last term ${c}({x_i})$ denotes the cost associated to the participation level of the MU, e.g., energy consumption and network bandwidth consumed. Similar to~\cite{chen2016incentivizing}, we assume that the cost is equal to $c{x_i}$, where $c$ is the MU's unit cost.\footnote{
 It is noted that, the same approach can also be applied to the model with heterogeneous unit cost (like ${c_i}{{x_i}^2}$) straightforwardly (Please refer to Appendix \ref{Sec:discussion on cost}). }
Then the utility of MU $i$ is expressed by:

\begin{equation}\label{Eq:2}
{u_i}({x_i},{{\bf{x}}_{ - {{i}}}}, {\bf{r}}) =  {a_i}{x_i} - {b_i}{x_i}^2 +\sum\limits_{j= 1}^N {{g_{ij}}{x_i}{x_j}}   + {r_i}{x_i} - {c}{x_i}.
\end{equation}

The monopoly CSP operates and maintains the platform with a fixed cost, which is ignored for the simplicity of the analysis later. Then, the formulation of revenue for the CSP is given by the payoff from total aggregated contribution of all MUs minus the total reward paid to MUs, i.e.,
\begin{equation}\label{Eq:3}
\Pi   = \mu\sum\limits_{i=1}^N {( s{x_i} - t{x_i}^2)} - \sum\limits_{i=1}^N {r_i}{x_i}.
\end{equation}
Similar to~\cite{xiong2017sponsor}, we also use the linear-quadratic function for tractability to transform the MUs' participation level into the monetary revenue of the CSP, which features the law of diminishing return. That is, an MU's contribution increases with the MU's effort level but the marginal return decreases. 
$\mu$ is an adjustable parameter representing the equivalent monetary worth of MUs' participation level, and $s, t> 0$ are coefficients capturing the concavity of the function.

We first address the incentive mechanism by modeling the strategic interactions between the CSP and the MUs as a two-stage single-leader multi-follower Stackelberg game.
\begin{definition}{Two-stage reward-participation game:}
\begin{itemize}
 \item Stage I (Reward): The CSP determines the reward, aiming at the highest revenue, i.e.,
   \[{{\bf{r}}^*} = \arg \max_{\bf{r}} \left\{\mu\sum\limits_{i=1}^N{( s{x_i} - t{x_i}^2)} - \sum\limits_{i=1}^N {r_i}{x_i}\right\};\]
 \item Stage II (Participation): Each MU $i \in {\cal N}$ chooses the participation level $x_i$, given the observed reward $\bf r$ and the participation levels of other MUs ${{\bf{x}}_{ - {{i}}}}$, with the goal to maximize its individual utility, i.e.,     \[x_i^* = \arg \mathop {\max_{x_i}}  {u_i}({x_i},{{\bf{x}}_{ - {{i}}}}, \bf{r}).\]
\end{itemize}
\end{definition}
We solve this two-stage Stackelberg game by finding a subgame perfect equilibrium for the cases of discriminatory incentive mechanism and uniform incentive mechanism for all MUs. 

\section{Stackelberg game equilibrium analysis with complete information}\label{Sec:Solution}
\subsection{Stage II: MUs' participation equilibrium}
Based on the definition of the Nash equilibrium, each MU chooses its participation level that is the best response. By setting the first-order derivative $\frac{{\partial {u_i}({x_i},{{\bf{x}}_{ - i}})}}{{\partial {x_i}}}$ to $0$, we obtain the best response of MU $i$ as follows:
\begin{equation}\label{Eq:4}
x_i^* =\max \left\{ 0, \frac{{{r_i} - c + {a_i}}}{{2{b_i}}} + \sum\limits_{j=1}^{N}\frac{{ {{g_{ij}}} }}{{2{b_i}}}{x_j} \right\}, \forall i.
\end{equation}
Each MU's best response includes two parts. $\frac{{{r_i} - c + {a_i}}}{{2{b_i}}} $ is independent from the strategies of the other MUs, and $\sum\limits_{j=1}^{N}\frac{{ {{g_{ij}}} }}{{2{b_i}}}{x_j}$ is dependent on the other MUs' participation levels due to underlying social network effects. Although the participation level strategy of each MU is obtained as in Eq.~(\ref{Eq:4}), the Nash equilibrium cannot be ensured to be unique or even exist since each MU may unboundedly increase its participation level if the other MUs' participation levels are large enough. Therefore, we present a sufficient assumption, under which there exists the unique Nash equilibrium as described in Theorem~1. Regarding the assumption, the MU has the upper bound on participation level, e.g., due to the battery capacity of a mobile device, and thus Assumption~1 is reasonable.
\begin{assumption}
$\sum\limits_{j=1}^{N} {\frac{{{g_{ij}}}}{{{2b_i}}}}  < 1, \forall i$.
\end{assumption}
\begin{theorem}
Under Assumption~1, the existence and uniqueness of MU participation equilibrium, i.e., the Nash equilibrium of Stage~II in this Stackelberg game, can be guaranteed.
\end{theorem}
\begin{proof}
\textbf{Existence of MU participation equilibrium:} We denote $\bf x^*$ as the strategy profile in the MU participation sub-game, and $x^\dag_i$ as the largest participation level in $\bf x^*$. Then, we have
\begin{eqnarray*}
x^\dag_i &=& \left(\frac{{{r_i} - c + {a_i}}}{{2{b_i}}} + \sum\limits_{j=1}^{N}\frac{{ {{g_{ij}}} }}{{2{b_i}}}{x_j} \right)^+ \le \frac{{{r_i} - c + {a_i}}}{{2{b_i}}} + \sum\limits_{j=1}^{N} {{x^\dag_i}\frac{{{g_{ij}}}}{{2{b_i}}}}\le \frac{{\left| {{{r_i} - c + {a_i}}} \right|}}{{2{b_i}}} + \sum\limits_{j=1}^{N} {{x^\dag_i}\frac{{\left| {{g_{ij}}} \right|}}{{2{b_i}}}}.
\end{eqnarray*}
Thus, under Assumption 1, we have
$$x^\dag_i  \le \frac{{\left|{{r_i} - c + {a_i}} \right|}}{{2{b_i} - \sum\limits_{j=1}^{N} {\left| {{g_{ij}} } \right|} }} =  \widehat x.$$
As a result, the strategy space $[0, \widehat x]$ is convex and compact, and the utility function ${u_i}({x_i},{\mathbf{x}}_{-i})$ is continuous in $x_i$ and ${\bf x}_{-i}$. We also have the second-order derivative of MU's objective function as follows:
$$\frac{{{\partial ^2}{u_i}}}{{{\partial ^2}{x_i}}}=-2b_i<0.$$ Thus, the MU participation sub-game is a concave game which always admits the Nash equilibrium.
\par
\textbf{Uniqueness of MU participation equilibrium:} Firstly, we have
$$ - \frac{{{\partial ^2}{u_i}}}{{{\partial}{x_i}^2}} =  - (-2{b_i} + {g_{ii}}) = 2{b_i}.$$ Then, based on Assumption~1, we have
\begin{eqnarray}
- \frac{{{\partial ^2}{u_i}}}{{{\partial}{x_i}^2}} &>& \sum\limits_{j=1}^N {{g_{ij}}} = \sum\limits_{j=1}^N {\left| {{g_{ij}}} \right|}  = \sum\limits_{j=1}^N{\left| { - \frac{{{\partial ^2}{u_i}}}{{\partial {x_i}{x_j}}}} \right|},
\end{eqnarray}
which satisfies the dominance solvability condition, i.e., Moulin's Theorem~\cite{moulin1984dominance}. As a result, the uniqueness of MU participation equilibrium is guaranteed under Assumption~1. The proof is then completed.
\end{proof}

Then, we propose the best response dynamics algorithm to obtain the Nash equilibrium with respect to MUs' participation level, as shown in Algorithm 1. The algorithm iteratively updates the MUs' strategies
based on their best response functions in Eq.~(\ref{Eq:4}), and converges to the Nash equilibrium of MU participation sub-game.
\begin{algorithm}\footnotesize
 \caption{Simultaneous best-response updating for finding Nash equilibrium of MU participation sub-game}
 \begin{algorithmic}[1]
 \STATE \textbf{Input:} \\
 Precision threshold $\epsilon$, $x_i^{[0]} \leftarrow 0 $, $x_i^{[1]} \leftarrow 1 + \epsilon$, $k\leftarrow 1$;
  \WHILE {$\left\|x_i^{[k]} - x_i^{[k-1]}\right\|_1>\epsilon$}
  \FORALL {$i \in \cal N$}
   \STATE  $x_i^{[k+1]} = \left( \frac{{{r_i} - c + {a_i}}}{{2{b_i}}} + \sum\limits_{j=1}^{N} {{x_j^{[k]}}\frac{{{g_{ij}} }}{{2{b_i}}}} \right)^+$;
  \ENDFOR
  \STATE $k\leftarrow k + 1 $;
  \ENDWHILE
  \STATE \textbf{Return} ${\bf{x}}_i^{[k]}$;
 \end{algorithmic}\label{algorithm}
\end{algorithm}

\begin{proposition}
Algorithm 1 achieves the Nash equilibrium of MU participation sub-game.
\end{proposition}
Note that Algorithm 1 achieves the approximate Nash equilibrium of MU participation sub-game, and the approximate accuracy, which measured by the gap between the achieved results and the optimal Nash equilibrium, depends on the precision threshold $\epsilon$. The convergence speed of the proposed algorithm also depends on precision threshold $\epsilon$. When $\epsilon$ is small, the number of iterations needed is large but the achieved results are more accurate. Conversely, when $\epsilon$ is big, the number of iterations needed is small but the achieved results are less accurate.

For ease of presentation, we have the following definitions, ${\bf B}:=diag(2b_1, 2b_2, \ldots, 2b_N)$, ${\bf a}:= [a_i]_{N \times 1}$, ${\bf 1}:= [1]_{N \times 1}$, ${\bf G}:=[g_{ij}]_{N \times N}$, ${\bf r}:= [r_i]_{N \times 1}$ and ${\bf{I}}$ is an $N \times N$ identity matrix. For the rest of the paper, similar to~\cite{xiong2017sponsor, zhou2017peer}, we consider the practical situation where all the MUs have positive participation levels at the Stackelberg equilibrium, i.e., a special case of Eq.~(\ref{Eq:4}). Then, with Lemma~1, we can rewrite Eq.~(\ref{Eq:4}) in a matrix form as follows:
\begin{equation}
{\bf{x}} = {\bf K}\left( {{\bf{a}} + {\bf{r}} - c{\bf{1}}} \right),
\end{equation}
where ${\bf K} = {\left( {{\bf{B}} - {\bf{G}}} \right)^{ - 1}}$.
\begin{lemma}
${{\bf{B}} - {\bf{G}}}$ is positive definite matrix, which is invertible.
\end{lemma}
\begin{proof}

We first denote $({\bf B-G})_{ij}$ as the value in the $i$th row and the $j$th column of the matrix ${\bf B-G}$, and it holds that ${\left( {{\bf{B}} - {\bf{G}}} \right)_{ii}} = 2{b_i} - {g_{ii}} =  2{b_i}$ since we have $g_{ii} = 0$. Under Assumption~1, we also have $2{b_i}  > \sum\limits_{j=1}^{N} {{g_{ij}}}$. Furthermore, we observe that $\sum\limits_{j = 1}^N {{g_{ij}}}  =  - \sum\limits_{j \ne i}^{} {{{\left( {{\bf{B}} - {\bf{G}}} \right)}_{ij}}}  = \sum\limits_{j \ne i}^{} {\left| {{{\left( {{\bf{B}} - {\bf{G}}} \right)}_{ij}}} \right|}  $. Therefore, it holds that ${\left( {{\bf{B}} - {\bf{G}}} \right)_{ii}} = 2{b_i} > \sum\limits_{j \ne i} {\left| {{{\left( {{\bf{B}} - {\bf{G}}} \right)}_{ij}}} \right|} $.

Accordingly, ${{\bf{B}} - {\bf{G}}}$ is strictly diagonally dominant and all the diagonal elements, i.e., $2b_i$ are larger than $0$. Based on Gershgorin circle theorem~\cite{weisstein2003gershgorin}, every eigenvalue $\lambda$ of ${{\bf{B}} - {\bf{G}}}$ satisfies
\begin{equation}
\left| {{{\left( {{\bf{B}} - {\bf{G}}} \right)}_{ii}} - \lambda } \right| < \sum\limits_{j=1}^{N} {\left| {{{\left( {{\bf{B}} - {\bf{G}}} \right)}_{ij}}} \right|}.
\end{equation}
Moreover, we know $\lambda > 0$, and thus ${{\bf{B}} - {\bf{G}}}$ is a positive definite matrix, from which its invertibility follows. The proof is then completed.
\end{proof}
\subsection{Stage I: Optimal incentive mechanism}
In this stage, the monopoly CSP determines the reward to be paid to the MUs, the objective of which is to maximize the CSP's revenue. Specifically, we investigate the discriminatory incentive mechanism and the uniform incentive mechanism, respectively. The significance of each incentive mechanism is as follows. Under the uniform incentive mechanism, the equilibrium ensures a fair reward applied to all MUs. Moreover, the uniform incentive mechanism is simple to implement in the crowdsensing applications. However, the CSP has limited degree of freedom to maximize its profit. By contrast, under the discriminatory incentive mechanism, the CSP can customize the reward for each MU, matching with the MU's preference and capability. As such, the profit obtained under the discriminatory incentive mechanism is expected to be superior to that of the uniform incentive mechanism. This is also confirmed in our numerical results.

\textit{1) Discriminatory incentive mechanism:}
Under reward discrimination, the CSP is able to provide different reward for different MUs as incentive to maximize its revenue. The revenue maximization problem can be formulated as follows:
\begin{equation}
\begin{aligned}
& \underset{\bf r}{\text{maximize}}
& & {\Pi} = \mu\sum\limits_{i=1}^{N}{( s{x_i} - t{x_i}^2)} - \sum\limits_{i=1}^{N} {r_i}{x_i}\\
& & & \quad = \mu( s{\bf{1}}^\top{\bf{x}}- {\bf{x}}^\top t{\bf{x}}) - {\bf{r}}^\top{\bf{x}}.\\
& \text{subject to}
& & {\bf{x}} = {\bf{K}}\left( {\bf{a}} + {\bf{r}} - c{\bf{1}}\right).\\
\end{aligned}\label{Eq:5}
\end{equation}
By plugging $\bf x$ into the objective function in Eq.~(\ref{Eq:5}), we have
\begin{equation}
\Pi = \mu \left(s{{\bf{1}}^ \top }{\bf{K}}\left( {{\bf{a}} + {\bf{r}} - c{\bf{1}}} \right) - t{\left( {{\bf{a}} + {\bf{r}} - c{\bf{1}}} \right)^ \top }{{\bf{K}}^2}\left( {{\bf{a}} + {\bf{r}} - c{\bf{1}}} \right)\right) - {{\bf{r}}^ \top }{\bf{K}}\left( {{\bf{a}} + {\bf{r}} - c{\bf{1}}} \right).
\end{equation}
Taking the partial derivative of the objective function in Eq.~(\ref{Eq:5}) with respect to the decision vector $\bf r$ to zero, i.e., $\frac{{\partial \Pi }}{{\partial {\bf{r}}}} = 0$, we obtain
\begin{equation}
\mu \left(s{\bf{K1}} - 2t{{\bf{K}}^2}\left( {{\bf{a}} + {\bf{r}} - c{\bf{1}}} \right)\right) - {\bf{K}}\left( {{\bf{a}} + {\bf{r}} - c{\bf{1}}} \right) - {\bf{Kr}} = 0.
\end{equation}
Then, we have
\begin{equation}
\mu \left(s{\bf{K1}} - 2t{{\bf{K}}^2}\left({{\bf{a}} - c{\bf{1}}} \right)\right) - {\bf{K}}\left( {{\bf{a}} - c{\bf{1}}} \right) = \left( {2{\bf{K}} + 2\mu t{{\bf{K}}^2}} \right){\bf{r}}.
\end{equation}
Finally, we obtain the optimal value ${\bf r}^*$, which is represented as follows:
\begin{equation}\label{Eq:6}
{\bf r}^* = {\left(2{\bf{I}}+ 2\mu t{\bf{K}}\right)^{ - 1}}\left( \mu \left(s{\bf{1}} - 2t{\bf{K}}\left( {{\bf{a}} - c{\bf{1}}} \right)\right) - \left( {{\bf{a}} - c{\bf{1}}} \right)\right).
\end{equation}
\textit{2) Uniform incentive mechanism:}
In this case, the CSP can only choose a single uniform reward to be paid to all the MUs, i.e., $r_i = r$, for all $i$. Then, the optimization problem is given by
\begin{equation}
\begin{aligned}
& \underset{r}{\text{maximize}}
& & {\Pi} = \mu\sum\limits_{i=1}^{N}{( s{x_i} - t{x_i}^2)} - {r}\sum\limits_{i=1}^{N}{x_i}\\
& & & \quad = \mu( s{\bf{1}}^\top{\bf{x}}- {\bf{x}}^\top t{\bf{x}}) - {{r}}{\bf{1}}^\top{\bf{x}}.\\
& \text{subject to}
& & {\bf{x}} = {\bf{K}}\left[{\bf{a}} + (r - c){\bf{1}}\right].\\
\end{aligned}\label{Eq:7}
\end{equation}
Similarly, we eliminate $\bf x$ from the objective function in Eq.~(\ref{Eq:7}), and we obtain
\begin{equation}
\Pi = \mu \left(s{{\bf{1}}^ \top }{\bf{K}}\left( {{\bf{a}} + (r - c){\bf{1}}} \right) - t{\left( {{\bf{a}} + (r - c){\bf{1}}} \right)^ \top }{{\bf{K}}^2}\left( {{\bf{a}} + (r - c){\bf{1}}} \right)\right) - r{{\bf{1}}^ \top }{\bf{K}}\left( {{\bf{a}} + (r - c){\bf{1}}} \right).
\end{equation}
Then, we evaluate its first-order optimality condition with respect to the reward $r$, which yields
\begin{equation}
\frac{{\partial \Pi }}{{\partial r}} = \mu \left( s{{\bf{1}}^ \top }{\bf{K1}} - 2t{\left( {{\bf{a}} + (r - c){\bf{1}}} \right)^ \top }{{\bf{K}}^2}{\bf{1}}\right)  - {{\bf{1}}^ \top }{\bf{K}}\left( {{\bf{a}} + (r - c){\bf{1}}} \right) - r{{\bf{1}}^ \top }{\bf{K1}} = 0.
\end{equation}
As a result, with simple steps, we obtain the optimal value of the uniform reward, which is represented by
\begin{equation}\label{Eq:8}
{r^*} = {\left( {2\mu t{{\bf{1}}^ \top }{{\bf{K}}^2}{\bf{1}} + 2{{\bf{1}}^ \top }{\bf{K1}}} \right)^{ - 1}}\left(\mu \left(s{{\bf{1}}^ \top }{\bf{K1}} - 2t{({\bf{a}} - c{\bf{1}})^ \top }{{\bf{K}}^2}{\bf{1}}\right) - {{\bf{1}}^ \top }{\bf{K}}\left( {{\bf{a}} - c{\bf{1}}} \right)\right) .
\end{equation}

Until now, we have obtained the optimal incentive mechanism in terms of uniform reward and discriminatory reward in closed-form solution with complete information, and hence validated the uniqueness of the Stackelberg equilibrium.

\section{Bayesian Stackelberg game theoretic analysis for socially-aware incentive mechanism with incomplete information}\label{Sec:Bayesian}
Recall from Section~\ref{Sec:Solution}, we assume that the MUs will truthfully report their personal information (type) to the CSP. This situation can happen when there exists a supervising entity in the market that is capable of monitoring, sharing and storing all behaviors to ensure that the MUs always report the correct information. However, without a supervising entity which is often the case in practice, the MU does not reveal private information (type) to the CSP because of the concern on privacy leakage or selfish behaviors. Therefore, the incomplete information scenario is more applicable to the real-world crowdsensing applications and address the incentive mechanism therein. In this section, we extend the analysis to the scenario where the social structure information (the social network effects) is not exactly known by the CSP and MUs. Thus, we formulate the incentive mechanism as a Bayesian Stackelberg game~\cite{duong2016stackelberg, chen2011optimal}, and evaluate the game equilibrium by defining and optimizing the expected utility of MUs and the expected revenue of the CSP.

\subsection{Problem formulation with social structure uncertainty}
In the model proposed in Section~\ref{Sec:Model}, the important social structure information may be uncertain or unknown by the decision makers, i.e., the CSP and MUs. Accordingly, this game can be modeled as a Bayesian game where the Bayesian analysis is adopted to predict the game outcome. In particular, the social relationship, i.e., social structure of each MU is private information and is considered as the type of the followers. Only its probability distribution is commonly known. Such distribution information can be obtained through, e.g., historical information or long-term learning.

The mobile social structure is represented by an interaction matrix, i.e., the adjacency matrix ${\mathbf{G}}$. As aforementioned, the element $g_{ij}$ denotes the strength of the influence of MU $j$ on MU $i$. Recall from Section~\ref{Sec:Model}, the utility of an MU can be expressed as follows:
\begin{equation}
{u_i}({x_i},{{\bf{x}}_{ - {{i}}}}, {\bf{r}}) =  {x_i} - \frac{1}{2}{x_i}^2 +\sum\limits_{j = 1}^N {{g_{ij}}{x_i}{x_j}}   + {r_i}{x_i} - {c}{x_i}.
\end{equation}

Note that we set $a_i = 1$ and $b_i = 1/2$ in order to concentrate on the social structure uncertainty. Moreover, without loss of generality, for all the social neighbours of MU $i$, i.e., $j \in {\cal N}_i$, $g_{ij} = \gamma > 0$, and $\gamma$ is a given social network effect coefficient. Thus, the above equation is rewritten as follows:
\begin{equation}
{u_i}({x_i},{{\bf{x}}_{ - {{i}}}}, {\bf{r}}) =  {x_i} - \frac{1}{2}{x_i}^2 +\gamma {x_i}\sum\limits_{j \in {\cal N}_i} {{x_j}}   + {r_i}{x_i} - {c}{x_i}.
\end{equation}
Therefore, the expected utility is expressed as follows:
\begin{equation}
{U_i}({x_i},{{\bf{x}}_{ - i}},{\bf{r}}) = {\mathbb E}\left[ {{u_i}({x_i},{{\bf{x}}_{ - i}},{\bf{r}})} \right] = {x_i} - \frac{1}{2}{x_i}^2 + \gamma {x_i}{\mathbb E}\left[ {\sum\limits_{j \in {\cal N}_i} {{x_j}} } \right] + {r_i}{x_i} - c{x_i}.
\end{equation}

The social structure leads to different in-degrees and out-degrees of MUs. The in-degree denotes the number of other MUs that a certain MU influences, the out-degree denotes the number of other MUs influencing this MU. Thus, the in-degree represents its influence and the out-degree represents its susceptibility. The distribution\footnote{Since the social ties/links are constructed by the in/out-degree information of MUs, the social ties/links are also treated as random variables in some sense. Note that $\gamma$ in the model is a given social network effect coefficient, which captures the strength of social ties/links. Although the value of $\gamma$ is given, the social ties/links follow certain probability distribution instead of being the constant value. It is also noteworthy that $\gamma$ can be treated as an approximate term instead of the accurate value. Nevertheless, the impacts of uncertainty of $\gamma$ can still be absorbed into the distribution of in/out-degree since they are interdependent.} of in-degree and out-degree captures the social network effects from the network interaction patterns~\cite{candogan2012optimal, bloch2013pricing, zhang2017optimal, fainmesser2016pricing, BelhajValue}. Note that the proposed model can still be applied to the asymmetric social ties, since we consider both the in-degree and out-degree distributions of each MU instead of the degree distribution. For example, an MU Alice has the social influence on another MU, but the latter may not have the social influence on Alice. The reason is that Alice may have different in-degree and out-degree.

The in-degree $l \in D$ and out-degree $k \in D$, where $D = \{0, 1, \ldots, k^{max}\}$ and $k^{max}$ denotes the maximum possible value. We define ${P}: D\to [0,1]$ and ${H} : D\to [0,1]$ as the probability distributions of out-degree and in-degree, respectively, and we have $\sum\limits_{k \in D} {P(k)}  = \sum\limits_{l \in D} {H(l)}  = 1$. Furthermore, we assume that two probability distributions are independent and their variances are denoted as ${\sigma _k}^2$ and ${\sigma _l}^2$, respectively. Due to consistency theory, we know $\sum\limits_{k \in D} {P(k)k} = \sum\limits_{l \in D} {H(l)l}=\overline k$, and thus $\overline k$ is referred to as the mean value of social network effects. 
Moreover, we have
\begin{equation}
{\mathbb E}\left[ {\sum\limits_{j \in {\cal N}_i}{{x_j}} } \right] = k_i \times {\rm Avg}({{\bf{x}}_{ - {{i}}}}),
\end{equation}
where ${\rm Avg}({{\bf{x}}_{ - {{i}}}}) = {\mathbb E}\left[ {{x_j}\left| j \in {\cal N}_i \right.} \right]$ denotes the average participation level of social neighbours of MU $i$.

In order to obtain the expression of ${\rm Avg}({{\bf{x}}_{ - {{i}}}})$, we employ the concept of ``\textit{Configuration Model}'' in \textit{Network Science}~\cite{newman2018networks} to model the random networks generated with only in-degree distribution. According to Configuration Model's property (See Chapter 12.2 in~\cite{newman2018networks}), to a user, the degree distribution of its randomly chosen neighbor is $\overline H (l) = {\frac{{H(l)l }}{{\sum\limits_{l' \in D} H(l')l' }}} $. In other words, a randomly selected social neighbours of MU $i$ has the in-degree distribution as $\overline H (l)$ and out-degree distribution as $P(k)$. Thus, by denoting the participation level of the MU with out-degree $k$ and in-degree $l$ as $x(k,l)$, we have~\cite{fainmesser2016pricing, BelhajValue}
\begin{equation}\label{Eq:BayesianAvg}
{\rm{Avg}}({{\bf{x}}_{ - i}}) = \sum\limits_{l \in D} {\bar H(l)\left( {\sum\limits_{k \in D} {P(k)x(k,l)} } \right)},
\end{equation}
where $\overline H (l) = {\frac{{H(l)l }}{{\sum\limits_{l' \in D} H(l')l' }}} $. Note that given ${\rm Avg}({{\bf{x}}_{ - {{i}}}})$, the participation level of MU $i$ only depends on the reward and its out-degree $k$. Thus, the final expected utility of MU $i$ is expressed as follows:
\begin{equation}\label{Eq:BayesainUtility}
{U_i}({x_i},{{\bf{x}}_{ - i}},{\bf{r}}, k_i) = (1 + {r_i} - c){x_i} - \frac{1}{2}{x_i}^2 + \gamma k_i{x_i}{\rm Avg}({{\bf{x}}_{ - i}}),
\end{equation}
and the type of the MU is its in-degree and out-degree, which is denoted as $(l, k)$.

Since only the distribution of the in-degree and out-degree information is known, instead of maximizing the revenue as defined in Eq.~(\ref{Eq:3}), the objective of the leader, i.e., the CSP, is to maximize its expected revenue, which is given as follows:
\begin{equation}\label{Eq:BayesianRevenue}
\Pi  = \sum\limits_{l \in D} {\left( {\sum\limits_{k \in D} {H(l)P(k)\left( {\left( {\mu s - r(k,l)} \right)x(k,l) - \mu t{{\left( {x(k,l)} \right)}^2}} \right)} } \right)},
\end{equation}
where $r(k,l)$ is the reward offered to the MU with out-degree $k$ and in-degree $l$.

\subsection{Stackelberg game equilibrium analysis}
We also adopt the backward induction to analyze the Bayesian Stackelberg game.

\textit{1) Follower game:}
For the given incentive or the reward determined by the CSP, we examine the Bayesain Nash equilibrium in the follower game which is characterized by the following theorem.

\begin{theorem}
The existence and uniqueness of the Bayesian Nash equilibrium of the follower game can be guaranteed, provided that the following condition
\begin{equation}
\gamma {k^{\max }} < 1
\end{equation}
holds.
\end{theorem}
\begin{proof}
\textbf{The existence of Bayesain follower game:} To prove that there exists at least one Bayesian Nash equilibrium in the follower game (Proposition 1 in~\cite{glaeser2000non}), we need to ensure that the following condition
\begin{equation}\label{Eq:BayesainCondition}
\frac{{\partial {U_i}({\overline x},{{\bf{x}}_{ - i}},{\bf{r}},{k_i})}}{{\partial {x_i}}} \le 0, \forall i \in {\cal N}, k \in Z^+, r \in {\mathbb R^+}, \exists \overline x  \ge 0, \forall x \le \overline x
\end{equation}
holds, where $\overline x = {\rm{Avg}}({{\bf{x}}_{ - i}})$. Since we have
\begin{eqnarray}
\frac{{\partial {U_i}(\bar x,{{\bf{x}}_{ - i}},{\bf{r}},{k_i})}}{{\partial {x_i}}} &=& (1 + {r_i} - c) - \bar x + \gamma {k_i}{\rm{Avg}}({{\bf{x}}_{ - i}}) \nonumber \\ &\le& (1 + {r_i} - c) - \bar x + \gamma {k^{\max }}\bar x \nonumber \\ &=& 1 + {r_i} - c + (\gamma {k^{\max }} - 1)\bar x,
\end{eqnarray}
we can ensure that the condition in Eq.~(\ref{Eq:BayesainCondition}) holds provided that $\gamma {k^{\max }} < 1$ is satisfied.

\textbf{The uniqueness of Bayesain follower game:} The proof of the uniqueness of the pure Bayesian Nash equilibrium can be directly derived from~\cite{glaeser2000non}. In particular, the sufficient condition that implies there exists at most one Bayesian Nash equilibrium is given as follows (Proposition 3 in~\cite{glaeser2000non}):
\begin{equation}\label{Eq:BayesianUnique}
\left| {\frac{{{\partial ^2}{U_i}(\bar x,{{\bf{x}}_{ - i}},{\bf{r}},{k_i})}}{{\partial {x_i}\partial {\rm{Avg}}({{\bf{x}}_{ - i}})}}\bigg /\frac{{{\partial ^2}{U_i}(\bar x,{{\bf{x}}_{ - i}},{\bf{r}},{k_i})}}{{\partial {x_i}\partial {x_i}}}} \right| < 1, \forall i \in {\cal N}.
\end{equation}
With simple steps, we have $\left| {\frac{{{\partial ^2}{U_i}(\bar x,{{\bf{x}}_{ - i}},{\bf{r}},{k_i})}}{{\partial {x_i}\partial {\rm{Avg}}({{\bf{x}}_{ - i}})}}/\frac{{{\partial ^2}{U_i}(\bar x,{{\bf{x}}_{ - i}},{\bf{r}},{k_i})}}{{\partial {x_i}\partial {x_i}}}} \right|= \left| \gamma {k_i} \right| \le \left| \gamma {k^{\max }} \right|$. Thus, if $\gamma {k^{\max }} < 1$ holds, the condition given in Eq.~(\ref{Eq:BayesianUnique}) is guaranteed. The proof is then completed.
\end{proof}

To obtain the closed-form expression of the unique Bayesian Nash equilibrium point in the follower game, we first apply partial derivative of the expected utility given in Eq.~(\ref{Eq:BayesainUtility}), i.e., $\frac{{\partial {U_i}(\bar x,{{\bf{x}}_{ - i}},{\bf{r}},{k_i})}}{{\partial {x_i}}} = 0$, as shown as follows:
\begin{equation}
x_i^* = 1 + {r_i} - c + \gamma {k_i}{{{\mathbb E}}}\left[ x_j {\left| j \in {\cal N}_i \right.}\right].
\end{equation}
Thus, we have
\begin{equation}
x(k,l) = 1 + r(k,l) - c + \gamma k{\mathbb E}\left[ {x(k,l)\left| {(k,l) \in {D^2}} \right.} \right].
\end{equation}
From Eq.~(\ref{Eq:BayesianAvg}), we have
\begin{eqnarray}\label{Eq:Analytical}
&&{\mathbb E}\left[ {x(k',l')\left| {(k',l') \in {D^2}} \right.} \right]=  \sum\limits_{l' \in D} {\overline H (l')\sum\limits_{k' \in D} {P(k')x(k',l')}} \nonumber \\ &=& \sum\limits_{l' \in D} {\left( {\overline H (l')\sum\limits_{k' \in D} {\left( {P(k')\left( {1 + r(k',l') - c + \gamma k'{\mathbb E}\left[ {x(k'',l'')\left| {(k'',l'') \in {D^2}} \right.} \right]} \right)} \right)} } \right)} \nonumber \\ &=& 1 + \overline r  - c + \gamma \overline k {\mathbb E}\left[ {x(k'',l'')\left| {(k'',l'') \in {D^2}} \right.} \right],
\end{eqnarray}
where $\overline r = \sum\limits_{l \in D} {\overline H (l)} \sum\limits_{k \in D} {P(k)} r(k,l)$ and $\overline k = \sum\limits_{l \in D} {\overline H (l)} \sum\limits_{k \in D} {P(k)} k = \sum\limits_{k \in D} {P(k)} k$. Since we also have
\begin{equation}
{\mathbb E}\left[ {x(k',l')\left| {(k',l') \in {D^2}} \right.} \right] = {\mathbb E}\left[ {x(k'',l'')\left| {(k'',l'') \in {D^2}} \right.} \right],
\end{equation}
it can be concluded from Eq.~(\ref{Eq:Analytical}) with the following expression
\begin{equation}
{\rm{Avg}}({{\bf{x}}_{ - i}}) = {\mathbb E}\left[ {{x_j}\left| {j \in {N_i}} \right.} \right] = \frac{{1 + \overline r  - c}}{{1 - \gamma \overline k }}.
\end{equation}
Therefore, we obtain the closed-form expression of the participation level of the MU with type $(k,l)$ in the Bayesian follower game, which is given as follows:
\begin{equation}\label{Eq:BayesianParticipation}
x^*(k,l) = 1 + r(k,l) - c + \gamma k\frac{{1 + \overline r  - c}}{{1 - \gamma \overline k }}.
\end{equation}

Note that Algorithm~1 can be implemented similarly in incomplete information scenario. The only difference is that the best-response function update policy in the $4$th line of Algorithm 1 is replaced by another update policy obtained from Eq.~(19). Since we have validated the existence and uniqueness of the Bayesian Nash equilibrium, the modified Algorithm 1 can achieve the Bayesian Nash equilibrium~\cite{han2012game}. Similar to that in complete information scenario, this Bayesian Nash equilibrium is also the approximate equilibrium due to the error $\epsilon$.

\textit{2) Leader game:}
As the CSP has the information on the degree distributions of MUs but has no information on MUs' type, and thus can offer only a uniform reward\footnote{Note that the uniform incentive mechanism is more applicable in incomplete information scenario, where the CSP has no information on the specific type of each individual MU. However, the in/out-degree distributions of MUs can be obtained through, e.g., historical information or long-term learning, which makes the uniform incentive mechanism feasible. This can also be confirmed by the closed-form solution for the optimal uniform reward, since the expression of the optimal reward only includes the mean and variance of the in/out-degree distributions instead of the type of individual MU.}, i.e., $r(k, l) = r$ for all MUs. The optimal incentive mechanism obtained from the leader game is characterized by the following theorem.

\begin{theorem}
The optimal reward offered by the CSP in the Bayesian Stackelberg game is unique, which is given as follows:
\begin{equation}
{r^*} = c - 1 + \frac{{\left( {\mu s + 1 - c} \right)\left( {1 - \gamma \overline k } \right)}}{{2\left( {1 - \gamma \overline k  + \mu t + \mu t{\gamma ^2}{\sigma _k}^2} \right)}}.
\end{equation}
\end{theorem}
\begin{proof}
Please refer to the appendix for the details.
\end{proof}

Furthermore, we study the benchmark case where the CSP knows both the in-degree and out-degree of any individual follower. In such a situation, the CSP is able to offer discriminatory reward as incentive, $r(k,l)$ for the MU with out-degree $k$ and in-degree $l$. Then, the revenue maximization problem faced by the CSP is characterized in the following theorem.
\begin{theorem}
Provided that the CSP clearly knows the type of each individual MU, the optimal discriminatory reward $r(k,l)$ offered to the MU with out-degree $k$ and in-degree $l$, is unique.
\end{theorem}
\begin{proof}
Please refer to the appendix for the details.
\end{proof}

\section{Performance evaluation}\label{Sec:Simulation}
In this section, we evaluate the performance of the proposed socially-aware incentive mechanisms of the CSP in crowdsensing applications, and investigate the impacts of different parameters of mobile networks on the performance.
\begin{figure}[t]
\centering
\begin{minipage}[t]{0.48\textwidth}
\centering
\includegraphics[width=6cm]{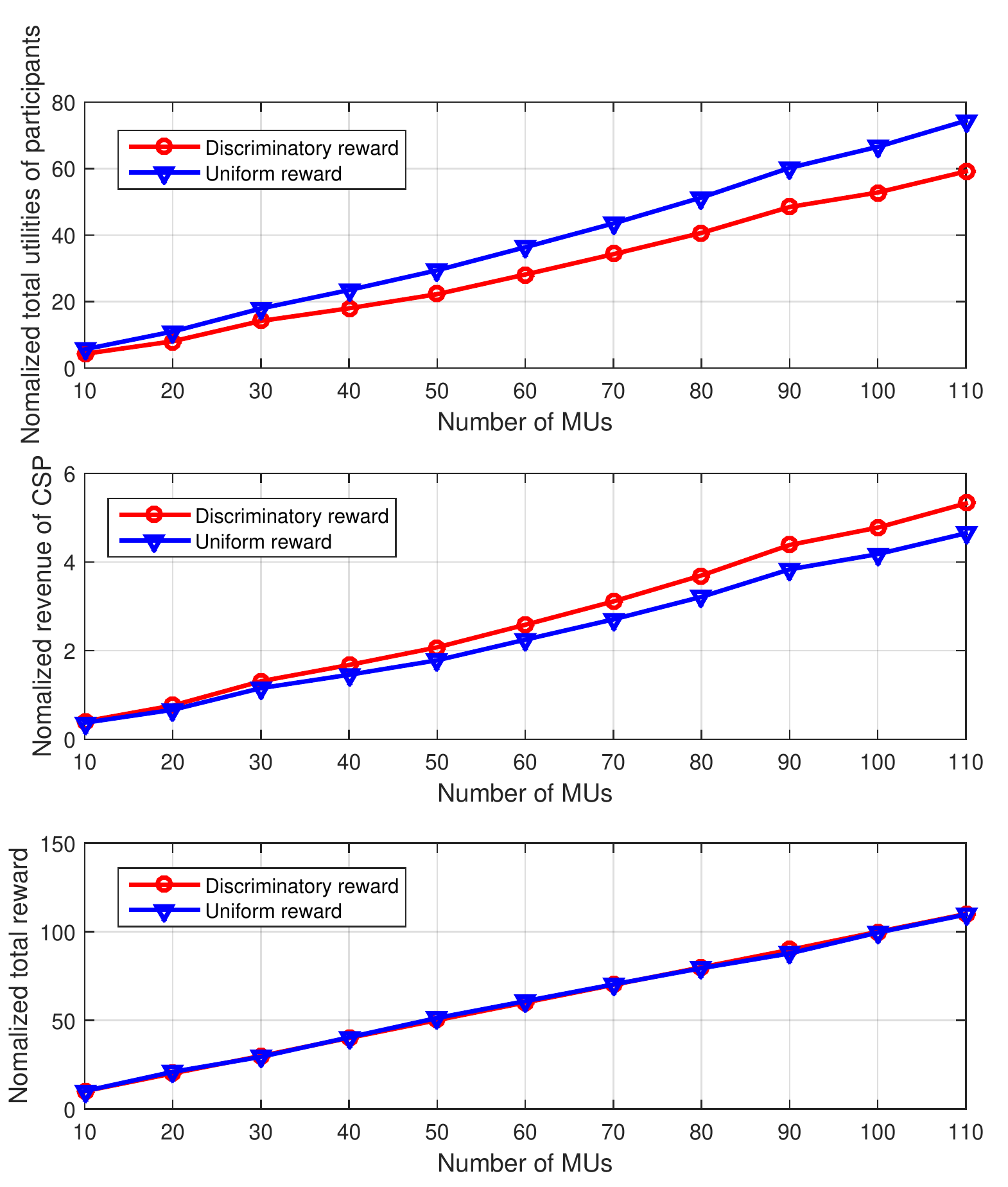}
\caption{The impact of total number of MUs on the crowdsensing service provider and mobile participants.}\label{Fig:number}
\end{minipage} \quad
\begin{minipage}[t]{0.48\textwidth}
\centering
\includegraphics[width=6cm]{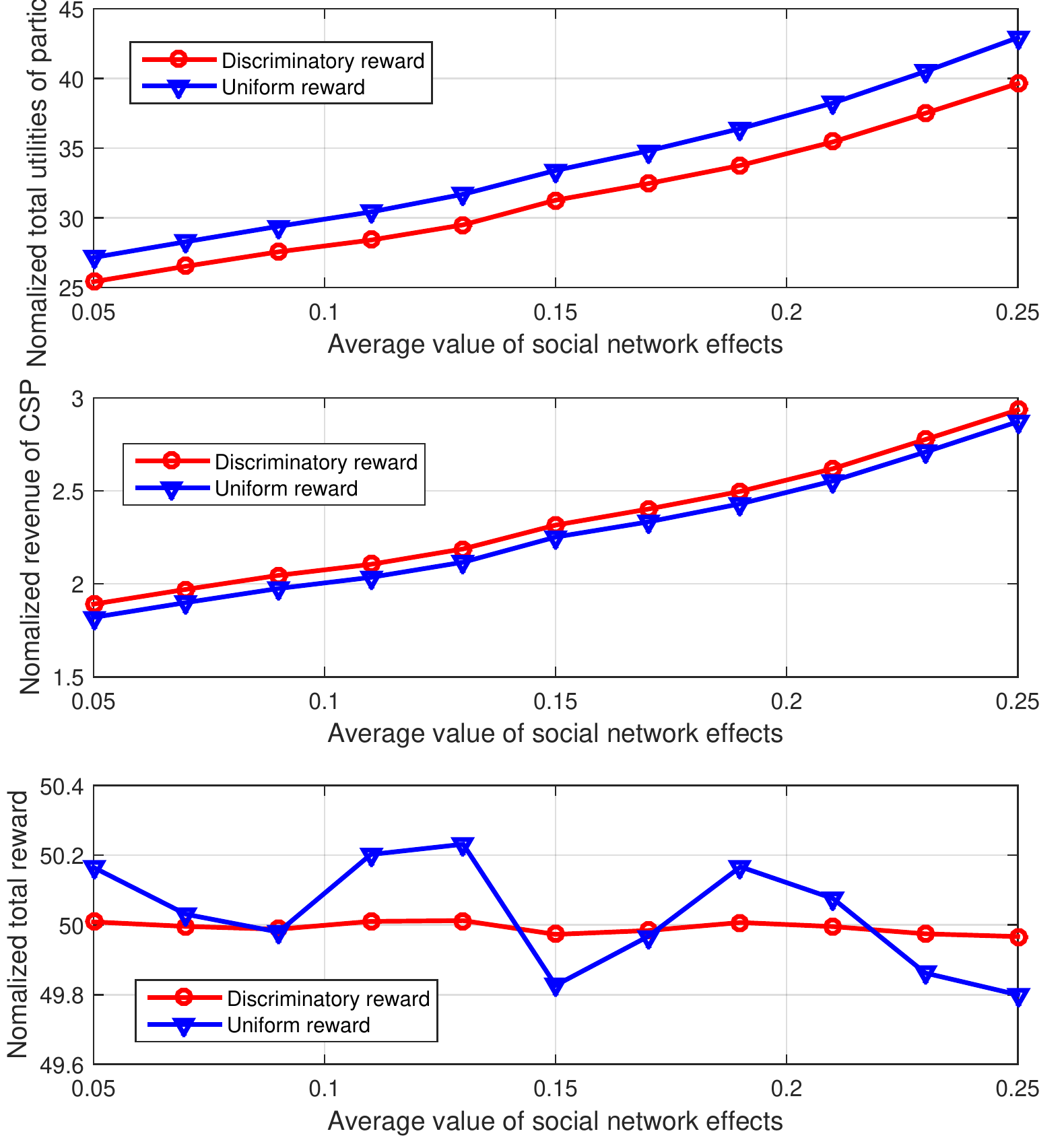}
\caption{The impact of average value of social network effects on the crowdsensing service provider and mobile participants.}\label{Fig:socialtie}
\end{minipage}
\end{figure}
\subsection{Investigation on Stackelberg game with complete information on social structure}
We consider a group of $N$ MUs, i.e., mobile participants, in a social network and set the parameters as follows. We assume the intrinsic parameters of MUs, i.e., $a_i$ and $b_i$ follow the normal distribution $\mathcal{N}(\mu_a, 2)$ and $\mathcal{N}(\mu_b,2)$. In addition, the social tie $g_{ij}$ between any two users $i$ and $j$ follows a normal distribution $\mathcal{N}(\mu_g, 1)$. The default parameters are set as: $c=15$, $\mu=0.1$, $s=20$, $t=0.05$, $\mu_a=\mu_b=15$, $\mu_g=0.05$ and $N=50$. Note that some of these parameters are varied according to the evaluation scenarios. As expected and verified in Fig.~\ref{Fig:number} and Fig.~\ref{Fig:socialtie}, the discriminatory incentive yields the larger revenue for the CSP, compared with the uniform incentive. Intuitively, the reason is that the CSP can adjust the reward according to individual MU's effort and contribution, which is proven by Fig.~\ref{Fig:index}. 

We next evaluate the impact of the total number of MUs on the proposed incentive mechanisms, as illustrated in Fig.~\ref{Fig:number}. As the number of MUs increases, the total utilities of participants and the revenue of the CSP also increase under both mechanisms. 
The reason is that when the total number of MUs increases, the number of social neighboring MUs also increases. Owing to the underlying social network effects, the MUs are motivated by their social neighbours to have higher participation levels, and the revenue of the CSP is improved accordingly. In addition, with the increase of total number of participants, the total offered reward increases since the CSP tends to encourage more MUs to participate, in order to attain a greater revenue gain. In particular, the discriminatory and uniform incentive mechanisms enable the CSP to reduce the reward paid to the MUs, i.e., the cost, and therefore achieve a greater revenue gain in turn. Figure~\ref{Fig:socialtie} depicts the impact of average value of social network effects on two entities of this network, i.e., the CSP and MUs. We observe that as the social network effects becomes stronger, the total utilities of participants and the revenue of the CSP also increase. Since when the strength of social tie is stronger, the additional benefits obtained from social network effects are greater. In other words, the socially-aware MUs are motivated by each other and have higher participation levels consequently. When the participation levels are high enough, the CSP is able to offer less reward to save money. In turn, the total utilities of participants and the revenue of the CSP are improved. Furthermore, we observe that the total offered reward under the uniform incentive mechanism and the discriminatory incentive mechanism have no big difference from both Figs.~\ref{Fig:number} and~\ref{Fig:socialtie}. The reason is that the CSP under the discriminatory incentive mechanism is able to achieve a balanced reward allocation with the similar cost. For example, the CSP can offer more reward to some MUs and less reward to some other MUs, which leads to a greater overall participation level. This intuition is demonstrated in Fig.~\ref{Fig:index}. From the third sub-figure in Fig.~3, we find that the uniform reward curve has several fluctuations suffering from the randomness of the social tie $g_{ij}$ when network effects become stronger.
Nevertheless, we can observe that the uniform reward still remains largely unchanged (around 50). This is different from the third sub-figure in Fig.~2, where we cannot observe the fluctuations. The reason is that Fig.~2 illustrates the impacts of the number of MUs on total offered reward from the CSP. Intuitively, the total reward increases when the number of participants increases. Thus, the slight fluctuations cannot be observed in Fig.~2 since the reward keeps increasing.

\begin{figure}[t]
\centering
\includegraphics[width=.6\textwidth]{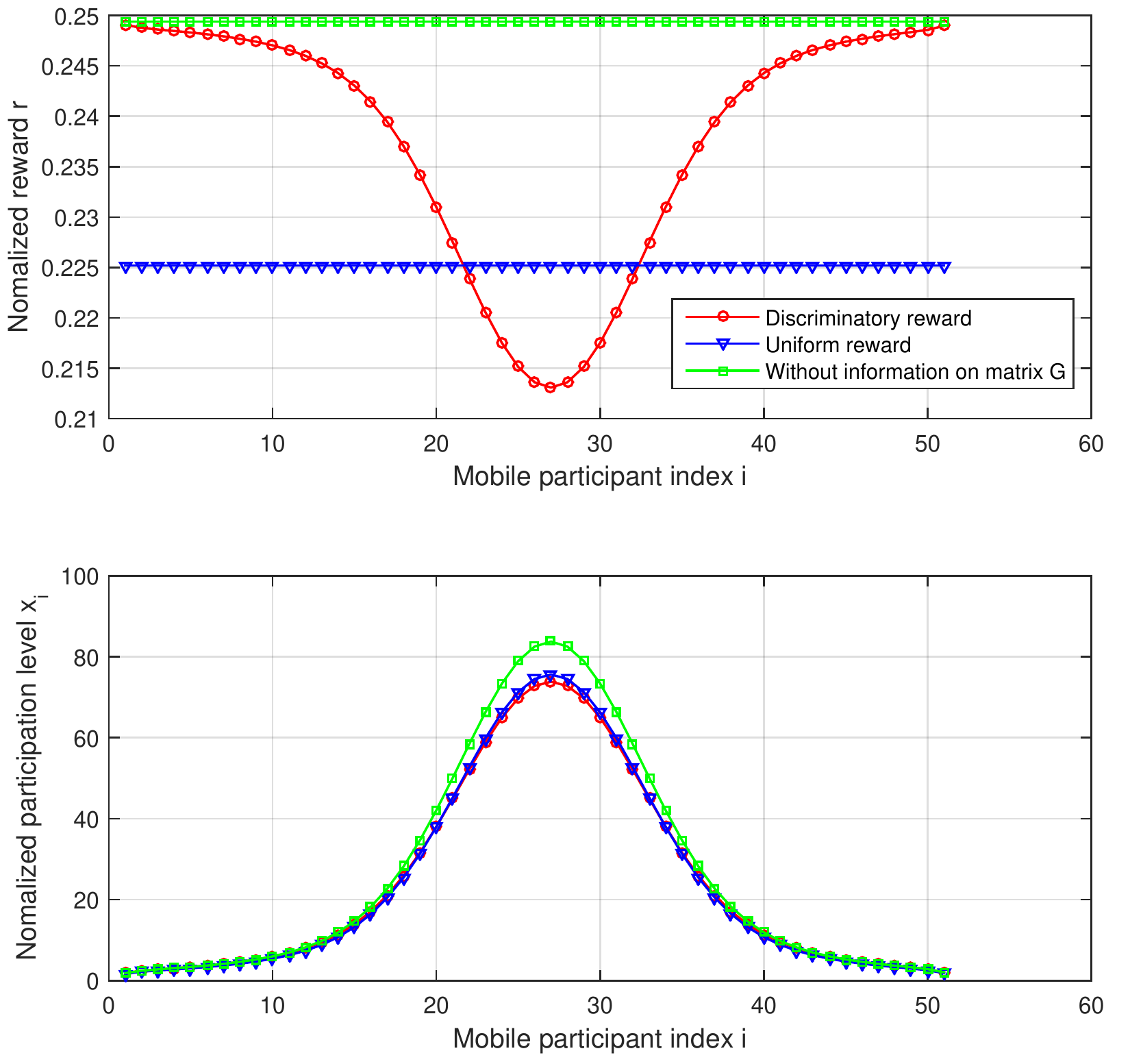}
\caption{A case illustration of distribution of normalized reward and participation level.}\label{Fig:index}
\vspace*{-8mm}
\end{figure}

Then, to explore the impacts of social network effects on each specific participant, we investigate the optimal reward and resulting MUs' participation level with the number of $50$ MUs, and we adopt the similar default parameters setting to that in above discussions. 
The adjacency matrix $G$ is generated as follows:
\begin{equation}\label{Eq:simulation}
\begin{cases}
   {g_{i,i + 1}} = 0.2\times\left(0.5 - {{\left(0.5 - \frac{{i - 1}}{N}\right)}^2}\right), &\mbox{$i \in [1,N - 1]$};\\
   {g_{i + 1,i}} = 0.2\times\left(0.5 - {{\left(0.5 - \frac{{i - 1}}{N}\right)}^2}\right), &\mbox{$i \in [1,N - 1]$};\\
   {g_{i,j}} = 0,&\mbox{otherwise}.
\end{cases}
\end{equation}

From Eq.~(\ref{Eq:simulation}), only participants who are adjacent in participant indexes (neighbours) can affect each other. From Fig.~\ref{Fig:index}, we observe that the CSP offers each participant the same reward when it has no information about the value of matrix $\bf G$. Given the reward from the CSP, the MUs have different participation level equilibrium as shown in Fig.~\ref{Fig:index}, where we observe that the participation levels of the MUs are socially related to each other. In particular, the $27$th MU is the most susceptible or influenced one in this network because it has the highest participation level given the same reward. On the contrary, the $1$st and the $51$st MUs are the most influential ones. Therefore, with the knowledge about the value of  matrix $\bf G$, the CSP is likely to offer more reward to the $1$st and the $51$st MUs and less to the $27$th MU, under the discriminatory incentive mechanism. The reason is that the CSP tends to have the highest participation level from the participant with the lowest cost and thus have a greater revenue gain. However, under the uniform incentive mechanism, the CSP can offer only the same reward to the participants. As such, the CSP usually offers more reward and promotes the participants to attain higher participation level, but the incurred extra cost is also very high, which decreases the revenue of the CSP consequently.
\begin{figure}[t]
\centering
\includegraphics[width=.6\textwidth]{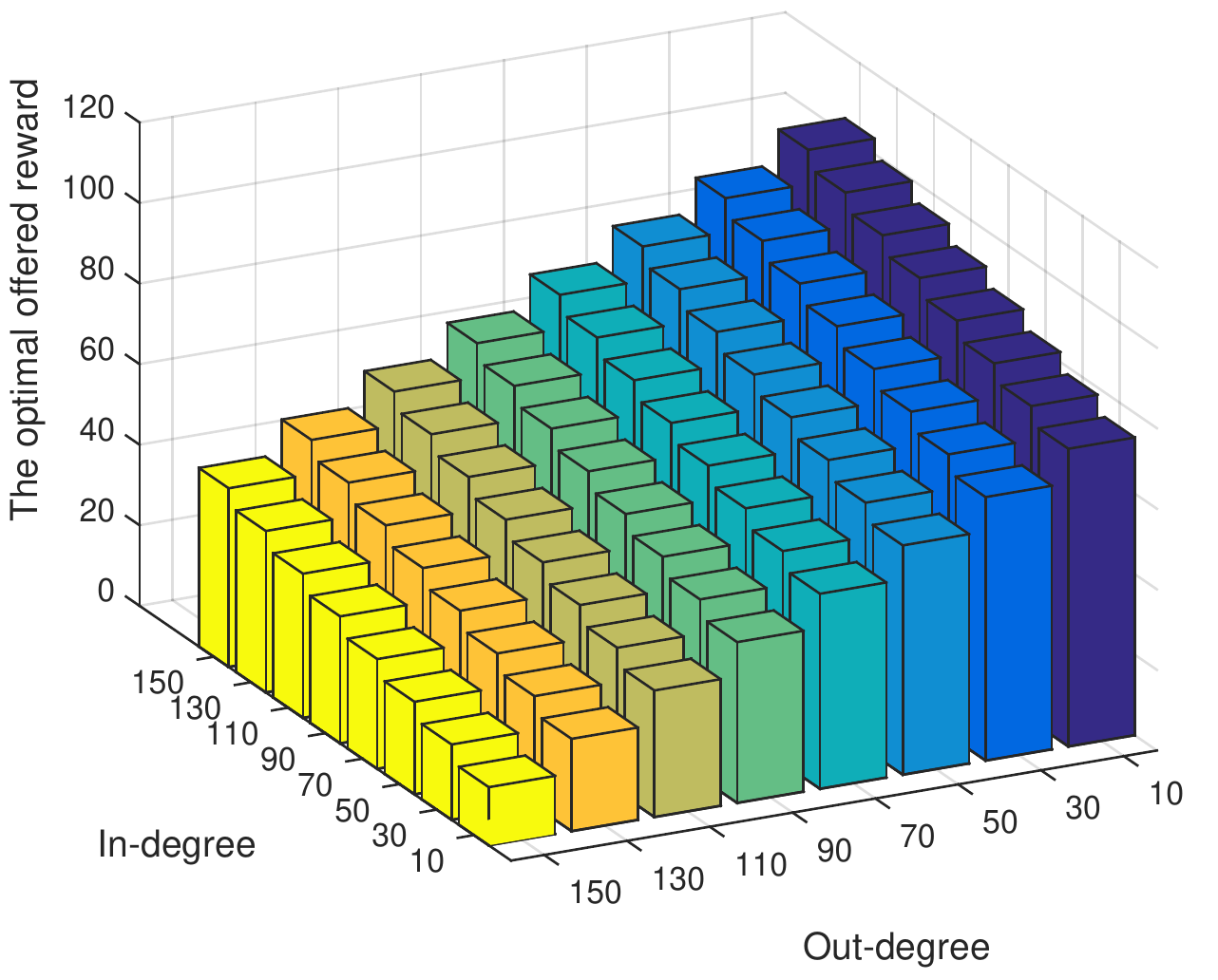}
\caption{The illustration of the optimal offered reward with respect to different in-degrees and out-degrees.}\label{Fig:rewardillustration}
\end{figure}
\subsection{Investigation on Bayesian Stackelberg game with incomplete information on social structure}
Similar to the above discussions, we consider a group of $N$ MUs. The in-degree and out-degree of MUs follow the normal distribution ${\cal N}(\overline k, \sigma_k^2 )$ and ${\cal N}(\overline k, \sigma_l^2)$, respectively. The parameters are set as follows: $\gamma = 0.01$, $\overline k = 20$, $\sigma_k^2 = \sigma_l^2 = 10$, $\mu = 10$, $s = 20$, $t = 0.05$, $c=15$, and $N = 100$.

We first study the optimal offered reward in terms of different in-degrees and out-degrees, as illustrated in Fig.~\ref{Fig:rewardillustration}. Interestingly, we find that the optimal offered reward increases with the increase of in-degree, and increases with the decrease of out-degree. Recall that the MU's in-degree represents its influence and the out-degree represents its susceptibility. As the in-degree of the MU increases, this MU can encourage more other MUs due to the underlying social network effects. In a crowdsensing-based road traffic information sharing platform, we can treat the drivers in critical central paths as the MUs with higher in-degree, i.e., greater influence. The road information from these drivers plays a great role, i.e., the participation of these drivers can greatly promote the participation of others. Thus, in the presence of social network effects, the CSP tends to offer more reward to the MUs with the higher in-degree, since they potentially motivate more participation level of other MUs. On the contrary, for the MUs with the higher out-degree, the CSP has no incentive to offer more reward. The reason is that the MUs with the higher out-degree are more susceptible, and these MUs are potentially positively affected by others. Consequently, the CSP is able to offer less reward to save the cost. 

\begin{figure}
 \begin{subfigure}[b]{0.32\textwidth}
 \centering
 \includegraphics[width=\columnwidth]{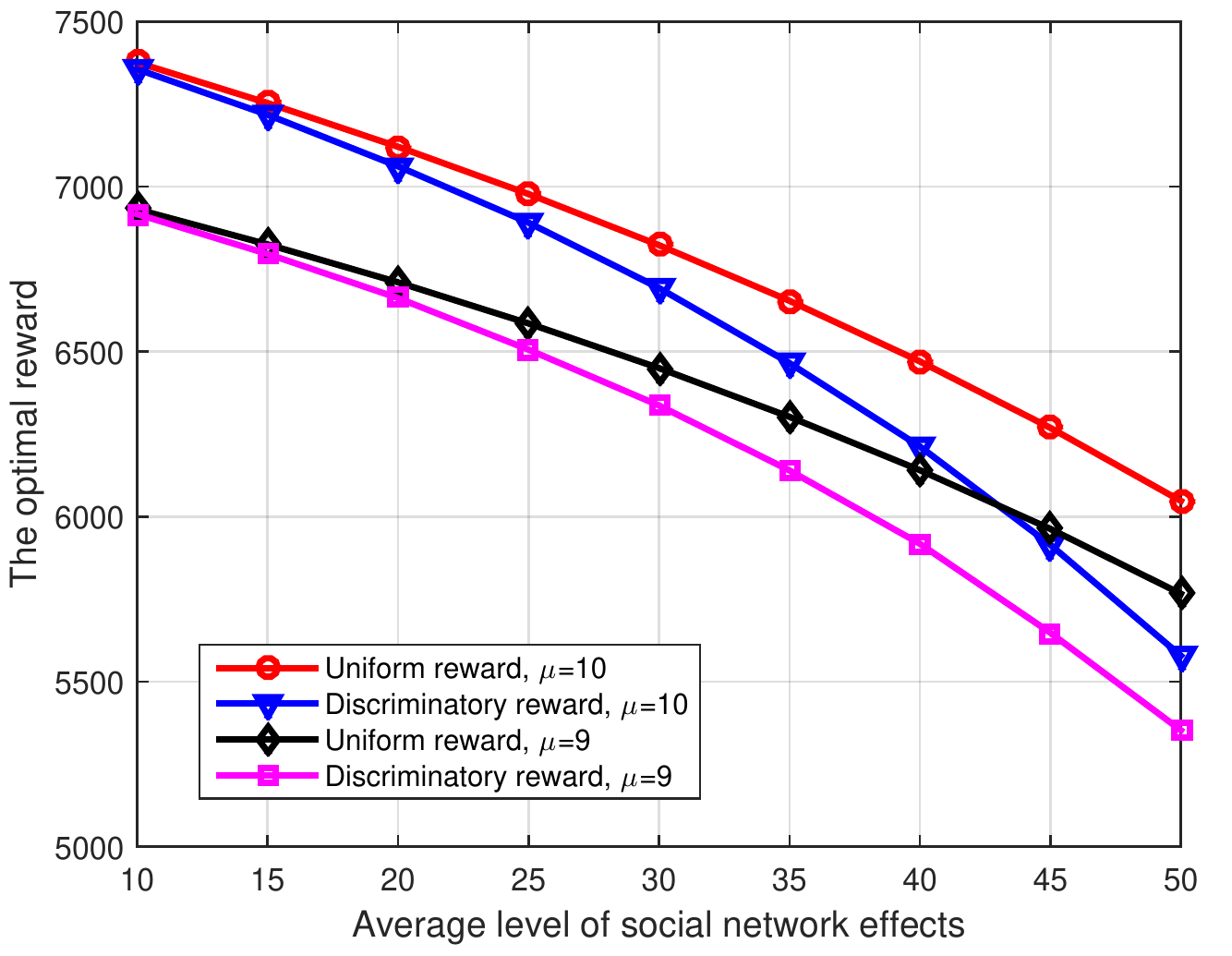}
 \caption{Offered reward}
 \label{Fig:reward_k}
 \end{subfigure}
 \begin{subfigure}[b]{0.32\textwidth}
 \centering
 \includegraphics[width=\columnwidth]{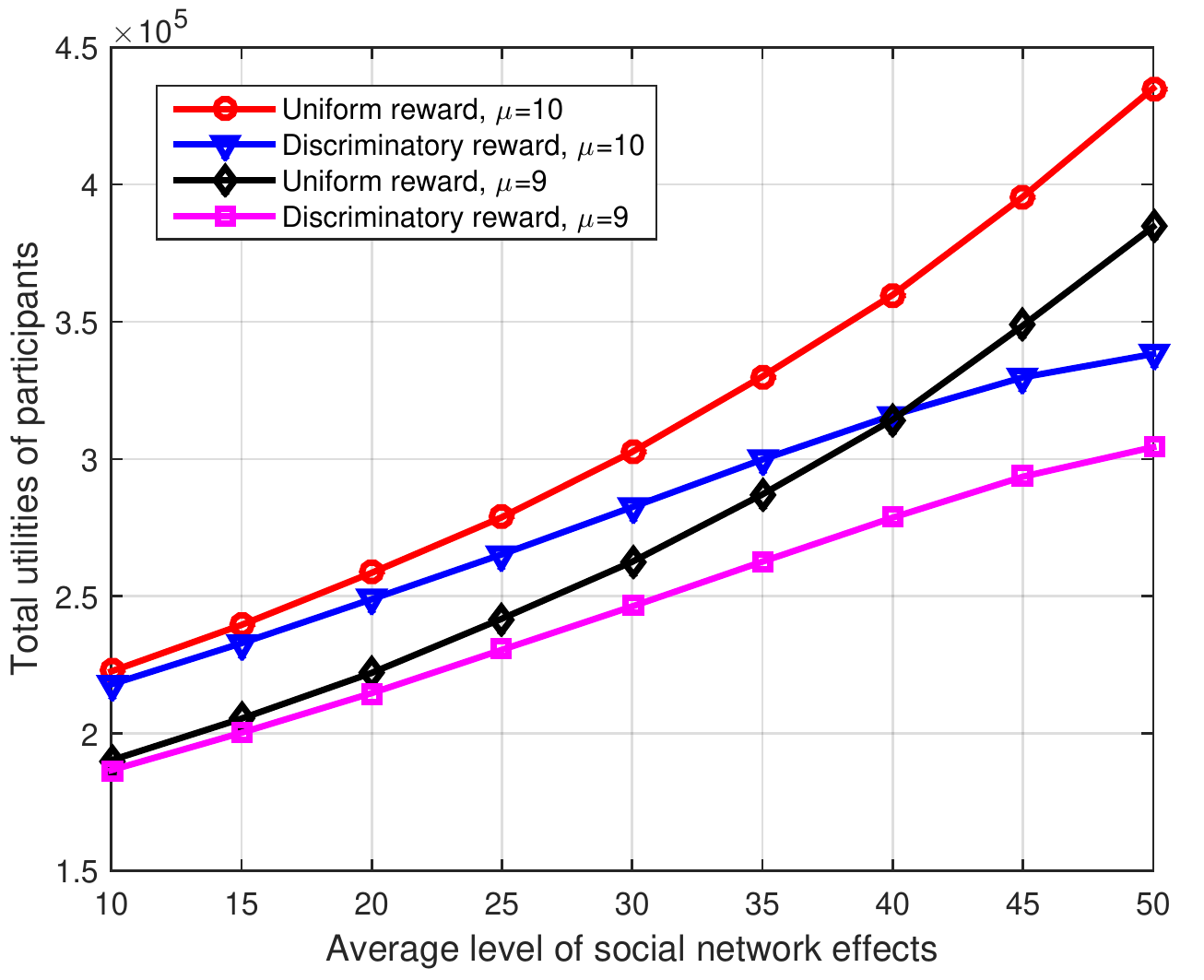}
 \caption{Total utilities of participants}
 \label{Fig:utility_k}
 \end{subfigure}
\begin{subfigure}[b]{0.32\textwidth}
\centering
\includegraphics[width=\columnwidth]{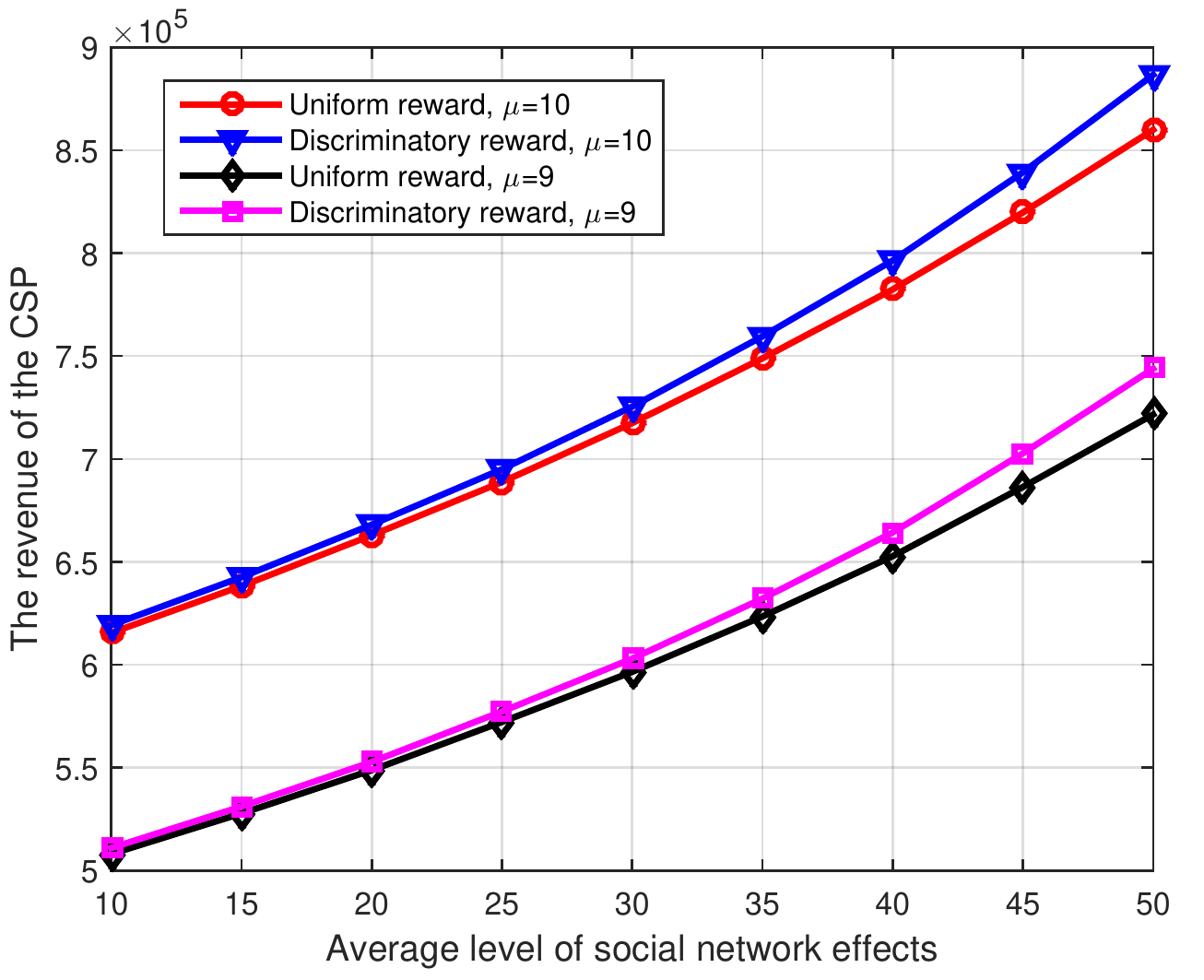}
 \caption{Revenue of the CSP}
 \label{Fig:revenue_k}
 \end{subfigure}
\centering
\caption{The impacts of mean value of social network effects.}
\label{Fig:AverageNetworkEffects}
\end{figure}

\begin{figure}
 \begin{subfigure}[b]{0.32\textwidth}
 \centering
 \includegraphics[width=\columnwidth]{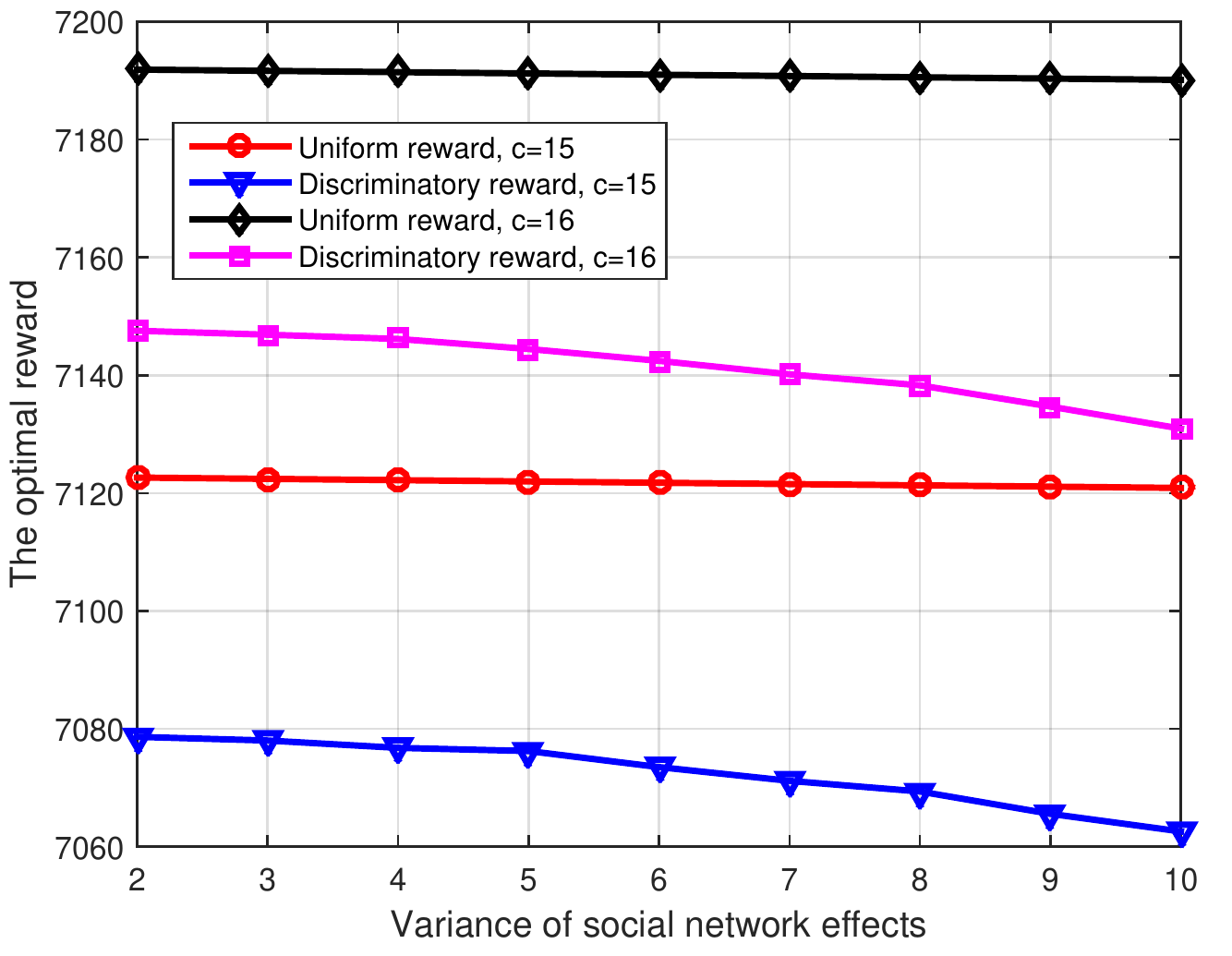}
 \caption{Offered reward}
 \label{fig:reward_cegma}
 \end{subfigure}
 \begin{subfigure}[b]{0.32\textwidth}
 \centering
 \includegraphics[width=\columnwidth]{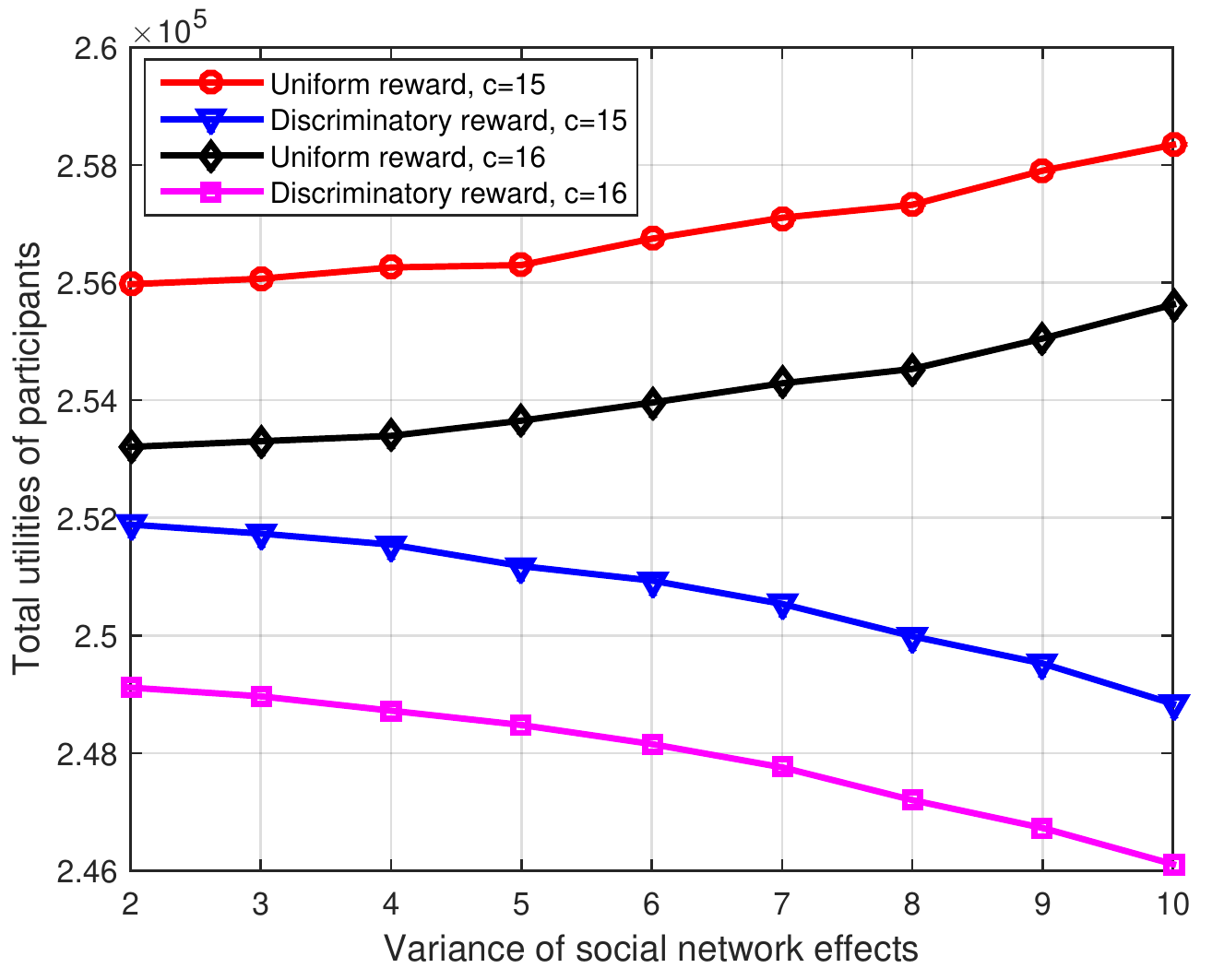}
 \caption{Total utilities of participants}
 \label{Fig:utility_cegma}
 \end{subfigure}
\begin{subfigure}[b]{0.32\textwidth}
\centering
\includegraphics[width=\columnwidth]{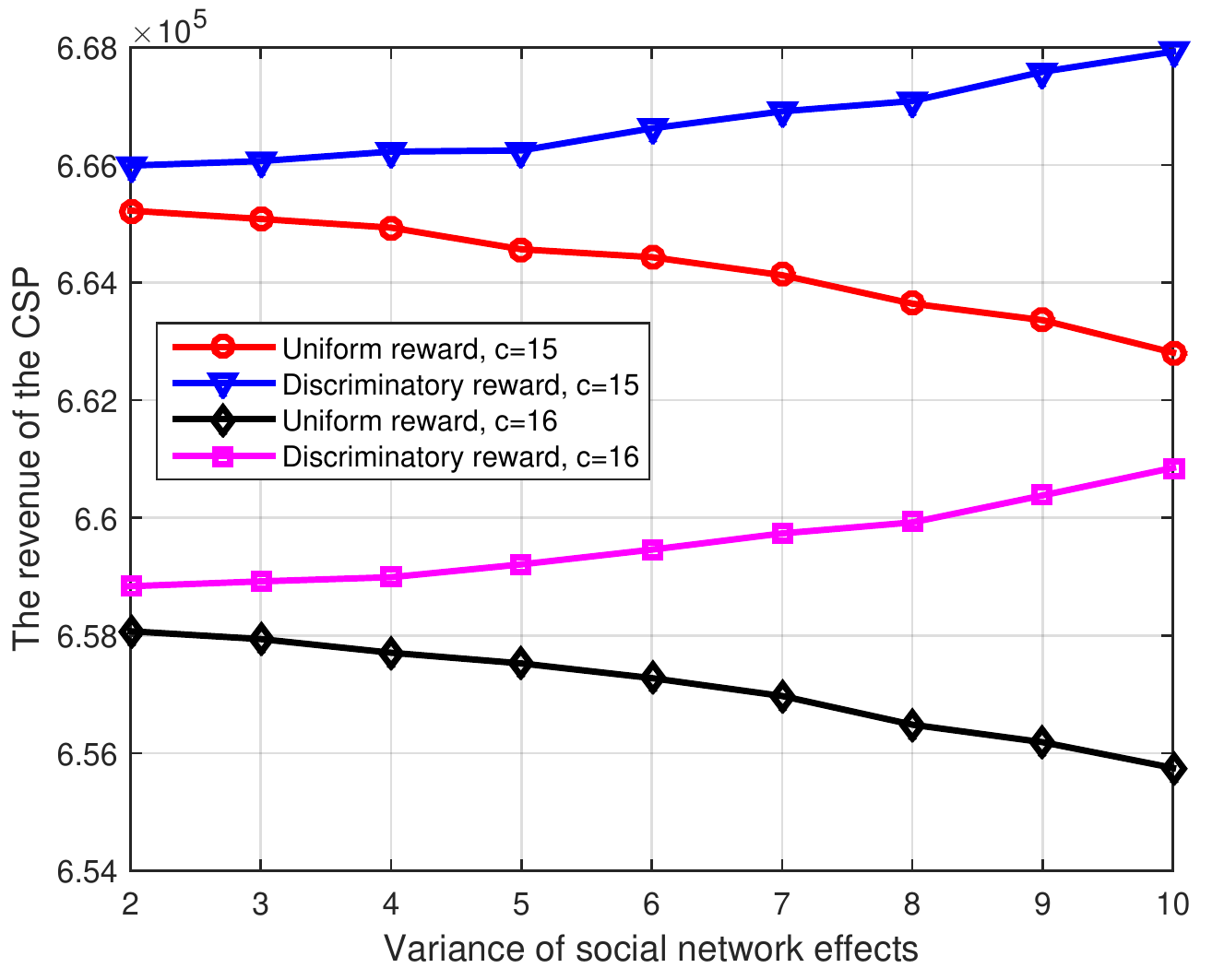}
 \caption{Revenue of the CSP}
 \label{Fig:revenue_cegma}
 \end{subfigure}
\centering
\caption{The impacts of variance of the distribution of in/out-degree.}
\label{Fig:VarianceNetworkEffects}
\end{figure}

Furthermore, we investigate the impacts of mean value of social network effects on the players of  Bayesian Stackelberg game in Fig.~\ref{Fig:AverageNetworkEffects}. As expected, we observe that the optimal offered reward decreases when the mean value of social network effects increases. The reason is that as the mean value of social network effects become stronger, the MUs can motivate each other to have higher participation level due to the interdependent participation behaviors. The total utilities of MUs become greater, and thus the CSP tends to offer less reward to save the cost. Consequently, the CSP achieves the greater revenue gain. In addition, comparing different curves with different value of $\mu$, we find that when $\mu$, i.e., the equivalent monetary worth of MUs' participation level increases, the CSP tends to offer more reward as the participation incentive. The reason is that the CSP is more inclined to arouse the enthusiasms of MUs when the CSP can transform the participation level of MUs to more monetary revenue efficiently. Therefore, to extract more surplus from MUs, the CSP offers more reward and thus achieves a greater revenue gain.

Figure~\ref{Fig:VarianceNetworkEffects} illustrates the impacts of variance of the distribution of social network effects on the players of Bayesian Stackelberg game. As the variance of social network effects decreases, the achieved revenue under uniform incentive mechanism is close to that under discriminatory incentive mechanism. The reason is that the heterogeneity of MUs is reduced when the value of variance decreases. We may consider an extreme case when the value of variance is zero, i.e., the MUs are homogeneous, the discriminatory incentive mechanism yields the same results as those of the uniform incentive mechanism. On the contrary, when the value of variance increases, the achieved revenue under the discriminatory incentive mechanism increases. The reason is that discriminatory incentive mechanism enables the CSP to exploit the different preference, i.e., parameters of utility function, for each of MUs, which leads to the decrease of total utilities of MUs and the increases of the revenue. Moreover, when the participation cost of MUs increases, the utility of MUs from participation is discounted. As such, the CSP tends to offer more reward to compensate the participation cost for improving their motivation. Meanwhile, the total utilities of MUs still decrease due to the increasing participation cost. Similarly, the revenue of the CSP decreases.

Lastly, both Figs.~\ref{Fig:AverageNetworkEffects} and~\ref{Fig:VarianceNetworkEffects} demonstrate the fact that the discriminatory incentive mechanism performs better in terms of the achieved revenue compared with uniform incentive mechanism. The intuition is that, with the certain social structure information, the CSP can set different reward for different MUs, as verified by the Fig.~\ref{Fig:rewardillustration}. As such, the CSP can significantly encourage the greater participation level of MUs. In particular, the social structure information guides the CSP to extract more surplus from the MUs' participation, which results in the greater revenue gain.

In summary, we draw the following engineering insights:
\begin{itemize}
\item The network effects tremendously stimulate higher mobile participation level, which leads to the greater total utilities of MUs as well as the greater revenue of the CSP.
\item The discriminatory incentive mechanism yields the greater revenue of the CSP compared with the uniform incentive mechanism in both complete and incomplete information scenarios.
\item The achieved revenue gap between the uniform and discriminatory incentive mechanisms depends on the variance of network effects.
\item The CSP has the incentive to offer more reward to the influential MUs and less reward to the susceptible MUs, in order to promote a greater overall participation level.
\end{itemize}

\section{Conclusion}\label{Sec:Conclusion}
In this work, we have developed a two-stage Stackelberg game theoretic model, and obtained the equilibrium using backward induction. The Crowdsensing Service Provider (CSP) determines the incentive in the first stage, and the Mobile Users (MUs) decide on their participation level in response to the observed incentive in the second stage. Taking the social (local) network effects among MUs into account, we have proposed two incentive mechanisms, i.e., discriminatory incentive and uniform incentive, where we have obtained the closed-form expression for optimal incentive. Moreover, we have formulated the Bayesian Stackelberg game to analyze the incentive mechanism, when the social network effects are uncertain. We have validated the existence and uniqueness of the Bayesian Stackelberg equilibrium by identifying the best response strategies of MUs. Performance evaluations have demonstrated that the network effects significantly improve the participation levels of MUs and the revenue of the CSP. Additionally, it has been confirmed that the social structure information helps the CSP to achieve greater revenue gain. The joint considerations of uncertainties of internal utility and social influence in the model are well worth studying in the future works.


\appendix
\subsection{Proof of Theorem 3:}
\begin{proof}
Similar to that in Section~\ref{Sec:Solution}, we first apply the unique Bayesian Nash equilibrium of the follower game given in Eq.~(\ref{Eq:BayesianParticipation}) into the objective function given in Eq.~(\ref{Eq:BayesianRevenue}). Since $r(k, l) = r$ for all MUs under the uniform incentive mechanism, we know that the unique participation level of the MU only depends on its out-degree $k$ from Eq.~(\ref{Eq:BayesianParticipation}), i.e., $x^*(k,l) = x^*(k)$. The expected revenue of the CSP given in~(\ref{Eq:BayesianRevenue}) is then expressed as follows:
\begin{equation}
\Pi  = \sum\limits_{k \in D} {P(k)\left( {\left( {\mu s - r} \right)x^*(k) - \mu t{{\left( {x^*(k)} \right)}^2}} \right)}.
\end{equation}
In particular, we have the following mathematical transformations,
\begin{eqnarray}
\Pi  &=& \left( {\mu s - r} \right)\left( {1 + r - c + \gamma \overline k \frac{{1 + r - c}}{{1 - \gamma \overline k }}} \right) - \mu t\sum\limits_{k \in D} {P(k){{\left( {1 + r - c + \gamma \overline k \frac{{1 + r - c}}{{1 - \gamma \overline k }}} \right)}^2}} \nonumber \\ &=& \left( {\mu s - r} \right)\frac{{1 + r - c}}{{1 - \gamma \overline k }} - \mu t\left( {1 + r - c} \right)^2 - 2\mu t\left( {1 + r - c} \right)\gamma \overline k \frac{{1 + r - c}}{{1 - \gamma \overline k }} \nonumber \\&&- \mu t\sum\limits_{k \in D} {P(k){k^2}{{\left( {\gamma \frac{{1 + r - c}}{{1 - \gamma \overline k }}} \right)}^2}} \nonumber \\ &=& \left( {\mu s - r} \right)\frac{{1 + r - c}}{{1 - \gamma \overline k }} - \mu t\left( {1 + r - c} \right)^2 - 2\mu t\left( {1 + r - c} \right)\gamma \overline k \frac{{1 + r - c}}{{1 - \gamma \overline k }} \nonumber \\ &&- \mu t\left( {{{\overline k }^2} + {\sigma _k}^2} \right){\left( {\gamma \frac{{1 + r - c}}{{1 - \gamma \overline k }}} \right)^2} \nonumber \\  &=& \left( {\mu s - r} \right)\frac{{1 + r - c}}{{1 - \gamma \overline k }} - \mu t{\gamma ^2}{\sigma _k}^2{\left( {\gamma \frac{{1 + r - c}}{{1 - \gamma \overline k }}} \right)^2} - \Bigg( \mu t\left( {1 + r - c} \right)^2 \nonumber \\ &&+ 2\mu t\left( {1 + r - c} \right)\gamma \overline k \frac{{1 + r - c}}{{1 - \gamma \overline k }} + \mu t{{\overline k }^2}{{\left( {\gamma \frac{{1 + r - c}}{{1 - \gamma \overline k }}} \right)}^2} \Bigg) \nonumber \\ &=& \left( {\mu s - r} \right)\frac{{1 + r - c}}{{1 - \gamma \overline k }} - \mu t{\gamma ^2}{\sigma _k}^2{\left( {\frac{{1 + r - c}}{{1 - \gamma \overline k }}} \right)^2} - \mu t{\left( {1 + r - c + \gamma \overline k \frac{{1 + r - c}}{{1 - \gamma \overline k }}} \right)^2} \nonumber \\ &=& \left( {\mu s - r} \right)\frac{{1 + r - c}}{{1 - \gamma \overline k }} - \mu t{\gamma ^2}{\sigma _k}^2{\left( {\frac{{1 + r - c}}{{1 - \gamma \overline k }}} \right)^2} - \mu t{\left( {\frac{{1 + r - c}}{{1 - \gamma \overline k }}} \right)^2} \nonumber \\ &=& \left( {\mu s + 1 - c - (1 + r - c)} \right)\frac{{1 + r - c}}{{1 - \gamma \overline k }} - \mu t{\gamma ^2}{\sigma _k}^2{\left( {\frac{{1 + r - c}}{{1 - \gamma \overline k }}} \right)^2} - \mu t{\left( {\frac{{1 + r - c}}{{1 - \gamma \overline k }}} \right)^2} \nonumber \\ &=& \left( {\mu s + 1 - c} \right)\frac{{1 + r - c}}{{1 - \gamma \overline k }} - \left( {1 - \gamma \overline k  + \mu t + \mu t{\gamma ^2}{\sigma _k}^2} \right){\left( {\frac{{1 + r - c}}{{1 - \gamma \overline k }}} \right)^2}.
\end{eqnarray}
Then, we evaluate its first-order optimality condition with respect to the reward, and we have
$\frac{{\partial \Pi }}{{\partial r}} = \frac{{\partial \Pi }}{{\partial \frac{{1 + r - c}}{{1 - \gamma \overline k }}}}\frac{{\partial \frac{{1 + r - c}}{{1 - \gamma \overline k }}}}{{\partial r}}$, which yields
\begin{equation}
\left( {\left( {\mu s + 1 - c} \right) - 2\left( {1 - \gamma \overline k  + \mu t + \mu t{\gamma ^2}{\sigma _k}^2} \right)\left( {\frac{{1 + r^* - c}}{{1 - \gamma \overline k }}} \right)} \right)\frac{1}{{1 - \gamma \overline k }} =0.
\end{equation}
Thus, we can conclude that
\begin{equation}
1 + r^* - c = \frac{{\left( {\mu s + 1 - c} \right)\left( {1 - \gamma \overline k } \right)}}{{2\left( {1 - \gamma \overline k  + \mu t + \mu t{\gamma ^2}{\sigma _k}^2} \right)}}.
\end{equation}
Therefore, the optimal uniform reward under Bayesian formulation is uniquely determined, which is given as follows:
\begin{equation}
{r^*} = c - 1 + \frac{{\left( {\mu s + 1 - c} \right)\left( {1 - \gamma \overline k } \right)}}{{2\left( {1 - \gamma \overline k  + \mu t + \mu t{\gamma ^2}{\sigma _k}^2} \right)}}.
\end{equation}
\end{proof}

\subsection{Proof of Theorem 4:}
\begin{proof}
The CSP determines $r(k,l)$ for the MU with out-degree $k$ and in-degree $l$ to maximize its expected revenue. The derivation of the optimal reward follows the similar steps discussed in Section~\ref{Sec:Solution}. We have the expected revenue of the CSP, which is expressed as follows:
\begin{multline}
\Pi =\sum\limits_{l \in D}\Bigg( \sum\limits_{k \in D}H(l)P(k)\bigg( \left( {\mu s - r(k,l)} \right)\Big( {1 + r(k,l) - c + \gamma k\frac{{1 + \overline r  - c}}{{1 - \gamma \overline k }}} \Big) - \\ \mu t\Big( 1 + r(k,l) - c + \gamma k\frac{1 + \overline r  - c}{1 - \gamma \overline k } \Big)^2 \bigg)  \Bigg).
\end{multline}
By taking the first derivative with respect to $r$ with any out-degree $m \in D$ and in-degree $n \in D$, we have
\begin{eqnarray}
&&\frac{{\partial \Pi }}{{\partial r(m,n)}}= \frac{\partial }{{\partial r(m,n)}}H(n)P(m)\Bigg( \left( {\mu s - r(m,n)} \right)\left( {1 + r(m,n) - c + \gamma m\frac{{1 + \overline r  - c}}{{1 - \gamma \overline k }}} \right) \nonumber \\&-& \mu t \left( {1 + r(m,n) - c + \gamma m\frac{{1 + \overline r  - c}}{{1 - \gamma \overline k }}} \right)^2 \Bigg)  + \frac{\partial }{{\partial r(m,n)}}\sum\limits_{l \ne n} \sum\limits_{k \ne m} H(l)P(k) \nonumber \\ &\times&\Bigg( \left( {\mu s - r(k,l)} \right)\left( {1 + r(k,l) - c + \gamma k\frac{{1 + \overline r  - c}}{{1 - \gamma \overline k }}} \right) \nonumber \\ &-& \mu t \left( {1 + r(k,l) - c + \gamma m\frac{{1 + \overline r  - c}}{{1 - \gamma \overline k }}} \right)^2 \Bigg).
\end{eqnarray}
Accordingly, we know
\begin{eqnarray}
&&\frac{{\partial \Pi }}{{\partial r(m,n)}} =  H(n)P(m)\Bigg( - \left( {1 + r(m,n) - c + \gamma m\frac{{1 + \overline r  - c}}{{1 - \gamma \overline k }}} \right) + \left( 1 + \gamma m\frac{\overline H (n)P(m)}{1 - \gamma \overline k } \right) \nonumber \\ &\times& \left( {\mu s - r(m,n)} \right)- 2\mu t\left( {1 + r(m,n) - c + \gamma m\frac{{1 + \overline r  - c}}{{1 - \gamma \overline k }}} \right)\left( {1 + \gamma m\frac{{\overline H (n)P(m)}}{{1 - \gamma \overline k }}} \right) \Bigg) \nonumber \\ &+& \frac{\partial }{{\partial r(m,n)}}\sum\limits_{k \ne m} \sum\limits_{l \ne n} H(l)P(k)\Bigg(\left( {\mu s - r(k,l)} \right)\gamma k\frac{{\overline H (n)P(m)}}{{1 - \gamma \overline k }} \nonumber \\ &-& 2\mu t\gamma k\frac{{\overline H (n)P(m)}}{{1 - \gamma \overline k }}\left( 1 + r(k,l) - c + \gamma m\frac{{1 + \overline r  - c}}{{1 - \gamma \overline k }} \right) \Bigg).
\end{eqnarray}
Since we know $\frac{{\partial \Pi }}{{\partial r(m,n)}} =0$, we conclude that
\begin{eqnarray}
0&=&H(n)P(m)\bigg( - \left( {1 + r(m,n) - c + \gamma m\frac{{1 + \overline r  - c}}{{1 - \gamma \overline k }}} \right) + \mu s - r(m,n) \nonumber \\&&- 2\mu t\left(1 + r(m,n) - c + \gamma m\frac{{1 + \overline r  - c}}{{1 - \gamma \overline k }} \right) \bigg)  + \gamma \frac{\overline H (n)P(m)}{{1 - \gamma \overline k }}\sum\limits_{l \in D} \sum\limits_{k \in D} H(l)P(k)k \nonumber \\ &&\times\left( \mu s - r(k,l) - 2\mu t\left(1 + r(k,l) - c + \gamma k\frac{{1 + \overline r  - c}}{1 - \gamma \overline k } \right) \right).
\end{eqnarray}
With simple steps, we obtain the following expression
\begin{eqnarray}\label{Eq:r(m,n)}
r^*(m,n) &=& \frac{1}{{2(1 + \mu t)}}\Bigg( c - 1 - \gamma m\frac{{1 + \overline r  - c}}{{1 - \gamma \overline k }} + \mu s - 2\mu t\left( {1 - c + \gamma m\frac{{1 + \overline r  - c}}{{1 - \gamma \overline k }}} \right) \nonumber \\ &&+ \frac{{\gamma n}}{{1 - \gamma \overline k }}\Bigg( {\mu s - 2\mu t\left( {1 - c} \right) - \frac{{2\mu t\gamma \left( {1 + \overline r  - c} \right)\left( {{\sigma _k}^2 + {{\overline k }^2}} \right)}}{{(1 - \gamma \overline k )\overline k }}} \Bigg) - \gamma n\frac{{1 + 2\mu t}}{1 - \gamma \overline k }\psi \Bigg), \nonumber \\
\end{eqnarray}
where
\begin{equation}\label{Eq:definition}
\psi  = \sum\limits_{l \in D} {\sum\limits_{k \in D} {\frac{1}{{\overline k }}kH(l)P(k)} } r(k,l).
\end{equation}

Moreover, based on the definition of $\overline r$, we have
\begin{eqnarray}
\overline r &=& \sum\limits_{m \in D} {\sum\limits_{n \in D} {\overline H (n)P(m)r(m,n)} } \nonumber \\  &=& \frac{1}{{2(1 + \mu t)}}\Bigg( c - 1 - \gamma \overline k \frac{{1 + \overline r  - c}}{{1 - \gamma \overline k }} + \mu s - 2\mu t\left( {1 - c + \gamma \overline k \frac{{1 + \overline r  - c}}{{1 - \gamma \overline k }}} \right) + \frac{{\gamma \overline k }}{{1 - \gamma \overline k }} \nonumber \\ &&\times\bigg( {\mu s - 2\mu t\left( {1 - c} \right) - \frac{{2\mu t\gamma \left( {1 + \overline r  - c} \right)\left( {{\sigma _k}^2 + {{\overline k }^2}} \right)}}{{(1 - \gamma \overline k )\overline k }}} \bigg) - \gamma \overline k \frac{{1 + 2\mu t}}{{1 - \gamma \overline k }}\psi  \Bigg).
\end{eqnarray}
Thus, we have
\begin{eqnarray}\label{Eq:overliner}
\overline r  &=& \frac{1}{{2(1 + \mu t)}}\Bigg(- \left( {1 - c + \gamma \overline k \frac{{1 + \overline r  - c}}{{1 - \gamma \overline k }}} \right)\left( {\frac{{2\mu t}}{{1 - \gamma \overline k }} + 1} \right) + \frac{{\mu s}}{{1 - \gamma \overline k }} \nonumber \\ &-& \frac{{2\mu t{\gamma ^2}{\sigma _k}^2\left( {1 + \overline r  - c} \right)}}{{{{\left( {1 - \gamma \overline k } \right)}^2}}}\left( {1 - c + \gamma \overline k \frac{{1 + \overline r  - c}}{{1 - \gamma \overline k }}} \right) - \gamma \overline k \frac{{1 + 2\mu t}}{{1 - \gamma \overline k }}\psi \Bigg).
\end{eqnarray}
Likewise, based on the definition of $\psi$ given in Eq.~(\ref{Eq:definition}), we know
\begin{equation}
\psi = \sum\limits_{l \in D} {\sum\limits_{k \in D} {\frac{1}{{\overline k }}kH(l)P(k)} } r(k,l).
\end{equation}
Thus, we have
\begin{eqnarray}\label{Eq:psi}
\psi &=& \frac{1}{{2(1 + \mu t)\bar k}}\Bigg( - \overline k \left( {1 - c} \right) - \frac{{\gamma \left( {1 + \overline r  - c} \right)\left( {{\sigma _k}^2 + {{\overline k }^2}} \right)}}{{1 - \gamma \overline k }} + \mu s\overline k  - 2\mu t\Bigg( \left( {1 - c} \right)\overline k  \nonumber \\
&& + \frac{{\gamma \left( {1 + \bar r - c} \right)\left( {{\sigma _k}^2 + {{\bar k}^2}} \right)}}{{1 - \gamma \bar k}}) + \frac{{\gamma {{\bar k}^2}}}{{1 - \gamma \bar k}}(\mu s - \frac{{2\mu t\gamma \left( {1 + \bar r - c} \right)\left( {{\sigma _k}^2 + {{\bar k}^2}} \right)}}{{(1 - \gamma \bar k)\bar k}}
\nonumber \\
&&-2\mu t\left( {1 - c} \right)\Bigg) - \gamma {{\overline k }^2}\frac{{1 + 2\mu t}}{{1 - \gamma \overline k }}\psi  \Bigg) \nonumber \\
&=& \frac{1}{{2(1 + \mu t)}}\Bigg(c - 1 - \frac{{\gamma \overline k \left( {1 + \overline r  - c} \right)\left( {{\sigma _k}^2 + {{\overline k }^2}} \right)}}{{1 - \gamma \overline k }} + \mu s - 2\mu t\left( {1 - c + \frac{{\gamma \overline k \left( {1 + \overline r  - c} \right)}}{{1 - \gamma \overline k }}} \right) \nonumber \\
&&+ \frac{{\gamma \overline k }}{{1 - \gamma \overline k }}\left( \mu s - 2\mu t\left( \left( {1 - c} \right) + \frac{{\gamma \overline k \left( {1 + \overline r  - c} \right)}}{{1 - \gamma \overline k }} \right) - \frac{{2\mu t\gamma \left( {1 + \overline r  - c} \right){\sigma _k}^2}}{{(1 - \gamma \overline k )\overline k }} \right) \nonumber \\&&- \frac{{\gamma {\sigma _k}^2\left( {1 + \overline r  - c} \right)}}{{(1 - \gamma \overline k )\overline k }} - \frac{{2\mu t\gamma {\sigma _k}^2\left( {1 + \overline r  - c} \right)}}{{(1 - \gamma \overline k )\overline k }} - \gamma \overline k \frac{{1 + 2\mu t}}{{1 - \gamma \overline k }}\psi  \Bigg) \nonumber \\
&=& \frac{1}{{2(1 + \mu t)}}\Bigg(- \left( {1 - c + \frac{{\gamma \overline k \left( {1 + \overline r  - c} \right)}}{{1 - \gamma \overline k }}} \right)\left( {1 + \frac{{2\mu t}}{{1 - \gamma \overline k }}} \right) + \frac{{\mu s}}{{1 - \gamma \overline k }} \nonumber \\ &&- \frac{{\gamma {\sigma _k}^2\left( {1 + \overline r  - c} \right)}}{{(1 - \gamma \overline k )\overline k }}\left( {1 + \frac{{2\mu t}}{{1 - \gamma \overline k }}} \right)-\gamma \overline k \frac{{1 + 2\mu t}}{{1 - \gamma \overline k }}\psi \Bigg).
\end{eqnarray}

The two expressions given in Eq.~(\ref{Eq:overliner}) and Eq.~(\ref{Eq:psi}) together formulate a full rank linear equation system with two variables, i.e., $\overline r$ and $\psi$. Thus, we can derive the closed-form expression for both $\overline r$ and $\psi$~\cite{kailath1980linear}. In particular, the closed-form expression for the variable $\overline r$ is obtained as
\begin{eqnarray}
\overline r  &=& {\rho ^{ - 1}}\frac{{\gamma \overline k }}{{2(1 + \mu t) - \gamma \overline k }}\left( { - \left( {1 - c + \frac{{\gamma \overline k \left( {1 - c} \right)}}{{1 - \gamma \overline k }}} \right)\left( {1 + \frac{{2\mu t}}{{1 - \gamma \overline k }}} \right) + \frac{{\mu s}}{{\gamma \overline k }}} \right) \nonumber \\ &&- \frac{{\mu t{\gamma ^2}{\sigma _k}^2\left( {1 - c} \right)}}{{{(1 - \gamma \overline k )}^2}(1 + \mu t)} + \frac{(1 + 2\mu t){\gamma ^2}{\sigma _k}^2\left( {1 - c} \right)}{2\left( {2 + 2\mu t - \gamma \overline k } \right)(1 - \gamma \overline k )(1 + \mu t)}\left( 1 + \frac{{2\mu t}}{{\gamma \overline k }} \right),
\end{eqnarray}
where
\begin{eqnarray}
\rho  &=& 1 + \frac{{\gamma \overline k }}{{2(1 + \mu t) - \gamma \overline k }}\left( {1 + \frac{{2\mu t}}{{1 - \gamma \overline k }}} \right) + \frac{{\mu t{\gamma ^2}{\sigma _k}^2\left( {1 + \overline r  - c} \right)}}{{{{(1 - \gamma \overline k )}^2}(1 + \mu t)}} \nonumber \\ &&- \frac{{(1 + 2\mu t){\gamma ^2}{\sigma _k}^2}}{{2\left( {2 + 2\mu t - \gamma \overline k } \right)(1 - \gamma \overline k )(1 + \mu t)}}\left( {1 + \frac{{2\mu t}}{{\gamma \overline k }}} \right).
\end{eqnarray}
Likewise, the closed-form expression for the variable $\psi$ can be derived through similar steps. Accordingly, $r(m,n)$ can be obtained after plugging these two closed-form expressions into Eq.~(\ref{Eq:r(m,n)}), and thus the solution of $r(m,n)$ is unique. The proof is then completed.
\end{proof}

\subsection{Discussions on asymmetric social influence:}\label{Sec:discussion on gij}

It is noted that our proposed approach is not restricted to the symmetry assumption that $g_{ij} = g_{ji}$. We will illustrate the model with asymmetric social influence in the complete information and incomplete information scenarios, respectively.

\textit{1) Complete Information on Social Structure:}

First, we will show how our approach can be extended to the model with asymmetric social influence in complete information scenario.

We first check whether Lemma 1 still holds when $g_{ij} \neq g_{ji}$. In Lemma~1, we proved ${{\bf{B}} - {\bf{G}}}$ is invertible by proving it is a positive definite matrix, and evidently the positive definite matrix is symmetric. Nevertheless, we find that ${{\bf{B}} - {\bf{G}}}$ is strictly diagonally dominant even when $g_{ij} = g_{ji}$ does not hold. According to Levy–Desplanques theorem~\cite{bailey1969bounds}, a strictly diagonally dominant matrix is non-singular. Thus, ${{\bf{B}} - {\bf{G}}}$ is still invertible. Therefore, without the symmetry assumption of $g_{ij}$, the Lemma~1 still holds.

Recall that ${\bf K} = {\left( {{\bf{B}} - {\bf{G}} } \right)^{ - 1}}$ is symmetric, since we have $g_{ij}= g_{ji}$. Thus, ${{\bf{K}}^ \top } + {\bf{K}}$ can be simplified as $2{\bf {K}}$. However, when $g_{ij}\neq g_{ji}$, ${\bf K }= {\left( {{\bf{B}}- {\bf{G}} } \right)^{ - 1}}$ is not symmetric, and ${{\bf{K}}^ \top } + {\bf{K}}$ cannot be simplified, which will complicate our mathematical calculation process.
 Nevertheless, the analytical solution for the profit maximization of the CSP is still structurally the same as that in Section IV-B.

For the discriminatory incentive mechanism: When $g_{ij}\neq g_{ji}$, the objective function of the CSP and the optimal value of ${\bf x}^*$ keep unchanged. By plugging $\bf x$ into the objective function of the CSP, we have
\begin{equation}
\Pi = \mu \left(s{{\bf{1}}^ \top }{\bf{K}}\left( {{\bf{a}} + {\bf{r}} - c{\bf{1}}} \right) - t{\left( {{\bf{a}} + {\bf{r}} - c{\bf{1}}} \right)^ \top }{{\bf{K}}^2}\left( {{\bf{a}} + {\bf{r}} - c{\bf{1}}} \right)\right) - {{\bf{r}}^ \top }{\bf{K}}\left( {{\bf{a}} + {\bf{r}} - c{\bf{1}}} \right).
\end{equation}

We next take the partial derivative of the objective function of the CSP with respect to the decision vector $\bf r$ to zero, i.e., $\frac{{\partial \Pi }}{{\partial {\bf{r}}}} = 0$, we have
\begin{equation}
\frac{{\partial \Pi }}{{\partial \bf r}} = \mu \left( {s{\bf{K1}} - t\left({{\bf{K}}^2} + {{\left({{\bf{K}}^2}\right)}^ \top }\right)\left( {{\bf{a}} + {\bf{r}} - c{\bf{1}}} \right)} \right) - {\bf{K}}\left( {{\bf{a}} + {\bf{r}} - c{\bf{1}}} \right) - {\bf{Kr}} = 0.
\end{equation}

Then, we have
\begin{equation}
\mu \left( {s{\bf{K1}} - t({{\bf{K}}^2} + {{({{\bf{K}}^2})}^ \top })\left( {{\bf{a}} - c{\bf{1}}} \right)} \right) - {\bf{K}}\left( {{\bf{a}} - c{\bf{1}}} \right) = \left( {2{\bf{K}} + \mu t({{\bf{K}}^2} + {{({{\bf{K}}^2})}^ \top })} \right){\bf{r}}.
\end{equation}

We obtain the optimal value ${\bf r}^*$, which is represented as follows:
\begin{equation}
{\bf{r}^*} = \left( {2{\bf{K}} + \mu t({{\bf{K}}^2} + {{({{\bf{K}}^2})}^ \top })} \right)^{-1} \left(\mu \left( {s{\bf{K1}} - t({{\bf{K}}^2} + {{({{\bf{K}}^2})}^ \top })\left( {{\bf{a}} - c{\bf{1}}} \right)} \right) - {\bf{K}}\left( {{\bf{a}} - c{\bf{1}}} \right)\right).
\end{equation}

For uniform incentive mechanism, ${\bf{x}} = {\bf{K}}\left[{\bf{a}} + (r - c){\bf{1}}\right]$. We obtain
\begin{equation}
\Pi  = \mu \left( {s{{\bf{1}}^ \top }{\bf{K}}\left( {{\bf{a}} + r{\bf{1}}} \right) - t{{\left( {{\bf{a}} + r{\bf{1}}} \right)}^ \top }{{\bf{K}}^2}\left( {{\bf{a}} + r{\bf{1}}} \right)} \right) - r{{\bf{1}}^ \top }{\bf{K}}\left( {{\bf{a}} + r{\bf{1}}} \right).
\end{equation}
Then, we evaluate its first-order optimality condition with respect to the reward $r$, which yields
\begin{equation}
\frac{{\partial \Pi }}{{\partial r}} = \mu \left( s{{\bf{1}}^ \top }{\bf{K1}} - 2t{\left( {{\bf{a}} + r{\bf{1}}} \right)^ \top }{{\bf{K}}^2}{\bf{1}}\right)  - {{\bf{1}}^ \top }{\bf{K}}\left( {{\bf{a}} + r{\bf{1}}} \right) - r{{\bf{1}}^ \top }{\bf{K1}} = 0.
\end{equation}
As a result, with simple steps, we obtain the optimal value of the uniform reward, which is represented by
\begin{equation}\label{Eq:8}
{r^*} = {\left(\mu t{{\bf{1}}^ \top }\left({{\bf{K}}^2} + {{\bf{K}}^{{2^ \top }}}\right) + 2{{\bf{1}}^ \top }{\bf{K1}}\right)^{ - 1}}\left( \mu \left(s{{\bf{1}}^ \top }{\bf{K1}} - t{({\bf{a}} - c{\bf{1}})^ \top }\left({{\bf{K}}^2} + {{\bf{K}}^{{2^ \top }}}\right){\bf{1}}\right) - {{\bf{1}}^ \top }{\bf{K}}\left( {{\bf{a}} - c{\bf{1}}} \right)\right).
\end{equation}

\textit{2) Incomplete Information on Social Structure:} In incomplete information scenario, the proposed solution is also not restricted to such symmetry assumption that $g_{ij} = g_{ji}$. In other words, social ties are not reciprocal under our consideration in incomplete information scenario. Instead, we consider that the adjacency matrix $\bf G$ leads to different in-degrees and out-degrees of MUs. The in-degree denotes the number of other MUs that a certain MU influences, the out-degree denotes the number of other MUs influencing this MU.  Thus, the in-degree represents its influence and the out-degree represents its susceptibility. The distribution of in-degree and out-degree captures the social network effects from the network interaction patterns. Recall that since we consider both the in-degree and out-degree distributions of each MU instead of the degree distribution, the proposed model can still be applied to the asymmetric social ties. For example, an MU Alice has the social influence on another MU, but the latter may not have the social influence on Alice. The reason is that Alice may have different in-degree and out-degree. 


\subsection{Discussion on cost function:}\label{Sec:discussion on cost}

Note that our model can be straightforwardly extended to a more generalized heterogenous unit cost and nonlinear cost function here, and we now show an example when the cost function becomes ${c_i}{{x_i}^2}$. We choose this function since it can represent the condition that the incurred cost increases when the MU keeps increasing its participation level, and the marginal increase is increasing (instead of keeping unchanged for linear cost unction).

In the following, we illustrate how our approach can be applied to the model with heterogenous non-linear cost function in the complete information and incomplete information scenarios, respectively.

\textit{1) Complete Information on Social Structure:}
Firstly, the utility of MU $i$ is reformulated as follows:
\begin{equation}
{u_i}({x_i},{{\bf{x}}_{ - {{i}}}}, {\bf{r}}) =  {a_i}{x_i} - {b_i}{x_i}^2 +\sum\limits_{j= 1}^N {{g_{ij}}{x_i}{x_j}}   + {r_i}{x_i} - {c_i}{{x_i}^2}.
\end{equation}

The best response of MU $i$ becomes:
\begin{equation}
x_i^* =\max \left\{ 0, \frac{{{r_i} + {a_i}}}{{2{b_i}} + {2{c_i}}} + \sum\limits_{j=1}^{N}\frac{{ {{g_{ij}}} }}{{2{b_i}}+ {2{c_i}}}{x_j} \right\}, \forall i.
\end{equation}

Then, we make an assumption for sufficient condition similar to that in Assumption 1, i.e., Assumption R1: $\sum\limits_{j = 1}^N {\frac{{{g_{ij}}}}{{2{b_i} + {c_i}}}}  < 1,\forall i$. Under such assumption, the MU has the upper bound on participation level, e.g., due to the battery capacity of a mobile device, and thus this assumption is reasonable. Provided that Assumption R1 holds, the existence and uniqueness of Nash equilibrium in MU participation game can still be proved in the way that is structurally the same as the proof of Theorem~1.

Next we can obtain the closed-form solution for the unique Nash equilibrium, which is written in matrix form as follows:
\begin{equation}
{\bf{x}} = {\bf K'}\left( {{\bf{a}} + {\bf{r}}} \right),
\end{equation}
where ${\bf K'}: = {\left( {{\bf{B}} + {\bf{C}} - {\bf{G}}} \right)^{ - 1}}$, ${\bf C}:=diag(2c_1, 2c_2, \ldots, 2c_N)$, ${\bf B}:=diag(2b_1, 2b_2, \ldots, 2b_N)$ and ${\bf G}:=[g_{ij}]_{N \times N}$.

Under Assumption~R1, we can prove ${{\bf{B}} + {\bf{C}} - {\bf{G}}}$ is invertible similar to the proof in Lemma~1. Next, we show how the incentive mechanism with heterogenous non-linear cost function is derived.

First, for the discriminatory incentive mechanism, the CSP's revenue is formulated as follows:
\begin{equation}
\Pi  = \mu \left( {s{{\bf{1}}^ \top }{\bf{K'}}\left( {{\bf{a}} + {\bf{r}}} \right) - t{{\left( {{\bf{a}} + {\bf{r}}} \right)}^ \top }{{\bf{K'}}^2}\left( {{\bf{a}} + {\bf{r}}} \right)} \right) - {{\bf{r}}^ \top }{\bf{K'}}\left( {{\bf{a}} + {\bf{r}}} \right).
\end{equation}

Taking the partial derivative of the above objective function with respect to the decision vector $\bf r$ to zero, we have
\begin{equation}
\frac{{\partial \Pi }}{{\partial \bf r}} = \mu \left( {s{\bf{K'1}} - 2t{{\bf{K'}}^2}\left( {{\bf{a}} + {\bf{r}} - c{\bf{1}}} \right)} \right) - {\bf{K'}}\left( {{\bf{a}} + {\bf{r}} - c{\bf{1}}} \right) - {\bf{K'r}} = 0.
\end{equation}

Then, we have
\begin{equation}
\mu \left( {s{\bf{K'1}} - 2t{{\bf{K'}}^2}{\bf{a}}} \right) - {\bf{K'a}} = \left( {2{\bf{K'}} + 2\mu t{{\bf{K'}}^2}} \right){\bf{r}}.
\end{equation}

Finally, we obtain the optimal value of ${\bf r}^*$, which is represented as follows:
\begin{equation}
{{\bf{r}}^*} = {\left( {2{\bf{I}} + 2\mu t{\bf{K'}}} \right)^{ - 1}}\left( {\mu \left( {s{\bf{1}} - 2t{\bf{K'a}}} \right) - {\bf{a}}} \right).
\end{equation}

Likewise, for uniform incentive mechanism, we have
\begin{equation}
\Pi  = \mu \left( {s{{\bf{1}}^ \top }{\bf{K'}}\left( {{\bf{a}} + r{\bf{1}}} \right) - t{{\left( {{\bf{a}} + r{\bf{1}}} \right)}^ \top }{{\bf{K'}}^2}\left( {{\bf{a}} + r{\bf{1}}} \right)} \right) - r{{\bf{1}}^ \top }{\bf{K'}}\left( {{\bf{a}} + r{\bf{1}}} \right).
\end{equation}
Then, we evaluate its first-order optimality condition with respect to the uniform reward $r$, which yields
\begin{equation}
\frac{{\partial \Pi }}{{\partial r}} = \mu \left(s{{\bf{1}}^ \top }{\bf{K'1}} - 2t{\left( {{\bf{a}} + r{\bf{1}}} \right)^ \top }{{\bf{K'}}^2}{\bf{1}}\right)  - {{\bf{1}}^ \top }{\bf{K'}}\left( {{\bf{a}} + r{\bf{1}}} \right) - r{{\bf{1}}^ \top }{\bf{K'1}} = 0.
\end{equation}
As a result, with simple steps, we can obtain the optimal value of the uniform reward, which is represented by
\begin{equation}
{r^*} = {\left( {2\mu t{{\bf{1}}^ \top }{{\bf{K'}}^2}{\bf{1}} + 2{{\bf{1}}^ \top }{\bf{K'1}}} \right)^{ - 1}}\left( \mu \left(s{{\bf{1}}^ \top }{\bf{K'1}} - 2t{{\bf{a}}^ \top }{{\bf{K'}}^2}{\bf{1}}\right) - {{\bf{1}}^ \top }{\bf{K'a}}\right) .
\end{equation}

\textit{2) Incomplete Information on Social Structure:}
Recall that the expected utility of MU $i$ in incomplete information scenario is reformulated as follows:
\begin{equation}
{U_i}({x_i},{{\bf{x}}_{ - i}},{\bf{r}}) = {\mathbb E}\left[ {{u_i}({x_i},{{\bf{x}}_{ - i}},{\bf{r}})} \right] = {x_i} - \frac{1}{2}{x_i}^2 + \gamma {x_i}{\mathbb E}\left[ {\sum\limits_{j \in {\cal N}_i} {{x_j}} } \right] + {r_i}{x_i} - c_i{x_i}^2.
\end{equation}

With the same deduction shown in Section V-A, the expected utility can be expressed as follows:
\begin{equation}
{U_i}({x_i},{{\bf{x}}_{ - i}},{\bf{r}}, k_i) = (1 + {r_i}){x_i} - (\frac{1}{2}+ c_i ){x_i}^2 + \gamma k_i{x_i}{\rm Avg}({{\bf{x}}_{ - i}}),
\end{equation}

Similar to the proof of Theorem~2, we also have
\begin{eqnarray}
\frac{{\partial {U_i}(\bar x,{{\bf{x}}_{ - i}},{\bf{r}},{k_i})}}{{\partial {x_i}}} &=& (1 + {r_i})  - c_i  \bar x- \bar x + \gamma {k_i}{\rm{Avg}}({{\bf{x}}_{ - i}}) \nonumber \\ &\le& (1 + {r_i}) - c_i  \bar x - \bar x + \gamma {k^{\max }}\bar x \nonumber \\ &=& 1 + {r_i} + (\gamma {k^{\max }} - 1 - c_i)\bar x.
\end{eqnarray}

Since we have the condition $\gamma {k^{\max }}< 1$, we know that $\gamma {k^{\max }} - c_i < 1$ holds evidently. Thus, we can still guarantee the existence of the Bayesian follower game. Moreover, it still holds that $\left| {\frac{{{\partial ^2}{U_i}(\bar x,{{\bf{x}}_{ - i}},{\bf{r}},{k_i})}}{{\partial {x_i}\partial {\rm{Avg}}({{\bf{x}}_{ - i}})}}/\frac{{{\partial ^2}{U_i}(\bar x,{{\bf{x}}_{ - i}},{\bf{r}},{k_i})}}{{\partial {x_i}\partial {x_i}}}} \right|= \left| \gamma {k_i} \right| \le \left| \gamma {k^{\max }} \right|$. Thus, the uniqueness of the Bayesian follower game is guaranteed even if we change the cost function of MUs. Based on the unique Bayesian Nash equilibrium, the CSP can still achieve the profit maximization through backward induction. The corresponding proof steps are structurally the same as the proof of Theorems~3 and~4 with the uniform cost.

\linespread{1.36}
\bibliography{bibfile}

\end{document}